\documentclass[11pt]{article}

\usepackage[margin=1.25in]{geometry}
\usepackage[T1]{fontenc}
\usepackage[utf8]{inputenc}
\usepackage{lmodern}
\usepackage{amsmath,amssymb,amsthm,mathtools}
\usepackage{bm}
\usepackage{algorithm}
\usepackage{algpseudocode}
\usepackage{tikz}
\usepackage{graphicx}
\usepackage{wrapfig}
\usepackage{subcaption}
\usepackage{booktabs}
\usepackage{caption}
\usepackage{enumitem}
\usepackage{comment}
\usepackage{etoolbox}
\usepackage{color,soul}
\usepackage[authoryear,round]{natbib}
\usepackage{hyperref}
\usepackage{cleveref}

\newtheorem{theorem}{Theorem}
\newtheorem{lemma}{Lemma}

\theoremstyle{definition}

\newtheorem{assumption}{Assumption}
\theoremstyle{remark}
\newtheorem{remark}{Remark}

\newcommand{\txth}{\tilde{x}_\theta}

\makeatletter
\AtBeginEnvironment{equation}{%
  \setlength{\abovedisplayskip}{4pt}%
  \setlength{\belowdisplayskip}{4pt}%
}
\AtBeginEnvironment{align}{%
  \setlength{\abovedisplayskip}{4pt}%
  \setlength{\belowdisplayskip}{4pt}%
}
\AtBeginEnvironment{align*}{%
  \setlength{\abovedisplayskip}{4pt}%
  \setlength{\belowdisplayskip}{4pt}%
}
\makeatother

\title{Diff2SP: Diffusion Models for Correlated Scenario Generation in Stochastic Programming}

\author{
  Haixiang Sun\\
  Purdue University\\
  \texttt{sun1321@purdue.edu}
  \and
  Andrew L. Liu\\
  Purdue University\\
  \texttt{andrewliu@purdue.edu}
}
\begin{document}

\maketitle

\begin{abstract}
Scenario generation is a critical component in stochastic programming (SP), as it directly influences the quality of decision-making under uncertainty. Existing approaches predominantly rely on either sampling-based techniques or supervised learning using neural networks. Sampling-based techniques often struggle to capture complex dependencies and rare but plausible events, while supervised learning requires fixed input-output pairs for training and is limited in its ability to generate a wide variety of realistic scenarios that are not restricted by predefined patterns or rules.  To address these limitations, we introduce Diff2SP, a diffusion-based generative framework that incorporates downstream optimization objectives directly into scenario generation. Unlike conventional methods that treat scenario generation and decision-making as separate steps, Diff2SP embeds stochastic optimization into the training process, enabling the generation of scenarios that are both statistically coherent and decision-aware. To formally justify this optimization-aware design, we establish a regret bounds that link distributional accuracy to decision quality, and establish sample complexity guarantees showing faster convergence than traditional generative models such as GANs. Empirical results on both synthetic and power-system datasets validate these theoretical insights, demonstrating that Diff2SP consistently improves both statistical fidelity and downstream optimization outcomes.
\end{abstract}



\section{Introduction}
\label{sec:intro}

Decision-making under uncertainty is a fundamental challenge across a wide range of applications, particularly in modern energy systems. As electric grids evolve through increasing renewable penetration, electrification, distributed energy resources, and rapidly growing electricity demand from AI infrastructure and data centers, system operators and planners must make decisions under increasingly complex and correlated sources of uncertainty.
Stochastic programming (SP) provides a principled approach to tackling such challenges by incorporating probabilistic models into optimization. In solving SP problems, scenario generation is a critical step~\citep{Mulvey1995, Ruszczynski2003, Hoyland2001, Pflug2012,Chopra2019ScenarioGI}, where a finite set of representative scenarios is constructed to approximate the underlying uncertainty. The quality of these scenarios directly influences the robustness and optimality of decision making~\citep{021814225969, Kim2015AGT}. 

Scenario generation is a long-established field with a broad and effective toolkit, and many existing methods represent correlated, non-Gaussian uncertainty well. The persistent difficulty is one of trade-offs among fidelity, diversity, and computational cost, compounded by the fact that scenarios are typically generated to match the data distribution rather than to serve the downstream decision they inform. Classical sampling-based methods, such as Monte Carlo and sample-average approximation~\citep{article12,Homem2014}, are asymptotically consistent, but they can require a large number of scenarios to represent the uncertainty accurately, particularly the rare, high-cost tail events that dominate risk-sensitive decisions, which makes them computationally demanding.

Scenario-reduction techniques~\citep{Heitsch2009, Pflug2012} reduce the size of a given scenario set under a probability metric, but their quality depends on the initial scenario pool. In practice, that pool is often generated from empirical or parametric models that impose stationarity, Gaussianity, or independence assumptions. Moment-matching schemes face a related limitation: by matching only finitely many low-order moments, they may leave higher-order dependence and multimodality unmodeled. A separate line of work trains neural networks to approximate solution-relevant quantities or predict high-quality solutions directly from problem parameters~\citep{patel2022neur2sp, dumouchelle2023neur2ro, wu2024hgcnsp}. These amortized predictors can be fast at inference, but they typically require training labels obtained by repeatedly solving optimization instances, learn a problem- and distribution-specific parameter-to-decision or parameter-to-value map, and do not produce a generative model of the uncertainty itself. As a result, they are not designed to generate diverse scenarios, and their performance under distribution shift depends on how well the learned surrogate generalizes beyond the training distribution.

To remove this dependence on labels and instead learn the uncertainty distribution directly, recent work has turned to unsupervised generative models.
Generative adversarial networks (GANs) were among the first applied to scenario generation~\citep{chen2018model, 8362314, DONG2022118387}, but their implicit, adversarial formulation is poorly matched to SP for two concrete reasons. First, GANs provide no explicit likelihood or density estimate, making it difficult to assess whether the generated scenarios faithfully represent tail events and cross-variable dependence. Second, adversarial training is often unstable: the min--max objective can lead to non-convergence, mode collapse, or vanishing gradients when the discriminator separates generated and real samples~\citep{Goodfellow2014, arjovsky2017wasserstein, salimans2016improved, Mescheder2018}.
Other generative models avoid adversarial training, but they introduce different structural restrictions. Variational autoencoders~\citep{Kingma2013} optimize an evidence lower bound under amortized Gaussian posteriors, which can over-smooth sharp tails and fine cross-variable dependence. Normalizing flows~\citep{Papamakarios2019, cramer2021principal} provide exact likelihoods, but require invertible, dimension-preserving maps with tractable Jacobians, forcing coupling- or autoregressive architectures that constrain how correlated, structured multivariate uncertainty can be represented.

Diffusion models~\citep{ho2020denoising,song2020denoising} offer a different set of trade-offs that are well suited to decision-aware scenario generation, where the generator must capture rare but decision-relevant events while remaining stable and differentiable enough to support optimization-guided training.
 By decomposing generation into many simple
Gaussian denoising steps, they replace the adversarial game with a single, score-matching objective and attain broad mode coverage without a
discriminator, making them less prone to dropping rare events. Just as important for our purposes, the generation process is differentiable end to end, allowing a decision-dependent training signal to be propagated through the sampler. We emphasize that diffusion is not a universal remedy: sampling is iterative and therefore more expensive at inference than a single forward pass, and sample quality depends on accurate training of the denoiser. Indeed, our own analysis (Section~\ref{sec:Complexity}) shows that the regret guarantees we establish do not require diffusion \emph{per se} -- any sufficiently regularized generator trained with a decision-aware objective could in principle attain them. What distinguishes diffusion is practical: its stable training objective and differentiable sampling process make it possible to incorporate a decision-dependent loss, typically nonconvex and high-variance, without the instabilities that make this unreliable for adversarial models. 

Motivated by this alignment rather than by the appeal of diffusion alone, we propose Diff2SP, a diffusion-based framework for scenario generation in stochastic programming. Diff2SP combines conditional diffusion models with transformer-based architectures to learn structured uncertainty distributions directly from data. By incorporating self-attention within the denoising process, the framework captures temporal dependencies and cross-variable correlations, while conditional generation enables adaptation to problem-specific contexts.

A distinguishing feature of Diff2SP is that the downstream stochastic optimization problem is incorporated into training through a differentiable optimization layer. Rather than learning scenarios solely to match the historical data distribution, the framework uses decision-aware feedback to guide scenario generation toward distributions that remain statistically representative while being relevant to the optimization task.

This paper makes three contributions. First, we develop a diffusion-based framework that embeds a stochastic optimization problem as a differentiable layer within scenario generation, thereby coupling uncertainty modeling with downstream decision making. Second, we establish regret guarantees that quantify how errors in the learned uncertainty distribution translate into losses in decision quality. Third, we derive sample complexity bounds for optimization-aware scenario generation and compare them with GAN-based learning, showing that directly incorporating decision objectives yields stronger theoretical guarantees. Extensive numerical experiments validate these theoretical insights and demonstrate improved decision performance in both synthetic and real-world settings.

The remainder of this paper is organized as follows. Section~\ref{sec_related works} reviews related work on scenario generation and optimization-integrated learning. Section~\ref{sec_diff2sp} introduces the proposed Diff2SP framework, detailing the diffusion-based architecture and optimization-aware training strategy. Section~\ref{sec:Complexity} provides theoretical analysis on sample efficiency and decision regret. Section~\ref{sec_experiment} presents empirical evaluations, including synthetic and real-world experiments, as well as ablation studies. Finally, Section~\ref{sec_conclusion} concludes the paper and outlines potential future directions.

\section{Related Works}
\label{sec_related works}
\textit{Optimization-guided and decision-focused learning.}
A growing body of work embeds constrained optimization problems directly into neural networks as differentiable layers~\citep{kotary2023predict}, with constructions for linear and integer programs~\citep{ferber2020mipaal, mandi2020interior, paulus2021comboptnet}, quadratic programs~\citep{amos2017optnet, bambade2024leveraging, pan2024bpqp}, cone programs~\citep{cvxpylayers2019}, and general convex programs~\citep{agrawal2019differentiating,blondel2022efficient, sun2022alternating, pineda2022theseus}. Gradients are propagated through the optimal solution via implicit differentiation of the optimality conditions~\citep{lorraine2020optimizing, blondel2022efficient}, enabling end-to-end pipelines that train a predictive model against a downstream decision~\citep{donti2017task,wang2019satnet}. This machinery underpins decision-focused learning, in which models are trained with decision regret rather than predictive loss so as to optimize decision quality directly~\citep{elmachtoub2022smart, mandi2024decision}, with decision-focused generative learning~\citep{wang2025gen} the most closely related to our setting. These approaches share our goal of aligning learning with downstream decisions, but they predominantly treat the uncertain parameters as point predictions or low-dimensional inputs and do not model the structured uncertainty -- temporal and cross-variable correlation, multimodality, and domain-specific constraints -- that characterizes real SP instances such as joint wind, solar, and electricity demand trajectories in energy applications. The decision-focused machinery alone therefore does not supply a generative model capable of producing diverse,
correlated scenarios.

\noindent\textit{Scenario generation.}
The scenario generation literature studies how to construct representative realizations of uncertain parameters for use in SP. As discussed in Section~\ref{sec:intro}, this literature includes classical sampling, moment-matching, and scenario-reduction methods, as well as more recent deep generative approaches based on GANs, VAEs, normalizing flows, and diffusion models. These methods differ substantially in their modeling assumptions, sampling procedures, and training objectives, but they generally separate uncertainty modeling from downstream optimization: scenarios are generated first and then supplied to the stochastic program. This separation complements, but does not resolve, the limitation of decision-focused learning discussed above. Decision-focused learning optimizes for downstream decision quality without a generative model of structured uncertainty, while scenario generation models uncertainty without explicitly incorporating the downstream decision. Diff2SP bridges this gap by embedding the SP problem as a differentiable optimization layer within the training of a conditional diffusion model, so that scenario generation is guided by both statistical fidelity and decision quality.

\section{Diff2SP -- The General Framework}
\label{sec_diff2sp}
Building on the perspective established above, we now present the Diff2SP framework. We first formalize the underlying stochastic optimization problem, then describe the diffusion-based scenario generator and the optimization-guided training procedure that links generated scenarios to downstream decision quality.
As a starting point, we consider the following stochastic optimization problem, hereafter denoted as \( SP(\xi) \):
\begin{equation}
\label{eq:SP}
F^*_{\xi}:=\min_{x \in \mathcal{X}} \ \mathbb{E}_{\xi} [f(x, \xi)],
\end{equation}
where $x \in \mathcal{X} \subset \mathbb{R}^n$ denotes the deterministic decision variables, $\mathcal{X}$ is the feasible set, $\xi \in \Xi$ is a random vector capturing the uncertainty in the problem parameters, and $F^*_{\xi}$ denotes the optimal objective function value corresponding to the random vector $\xi$. In general, the feasible set $\mathcal{X}$ may itself depend on the realization of the uncertainty $\xi$, in which case we write $\mathcal{X}(\xi)$ to emphasize this dependence.

To introduce the Diff2SP framework, we first describe the core diffusion mechanism, then present the transformer-based architecture used to capture structural correlations, and finally detail the optimization-guided training procedure that aligns scenario generation with downstream decision quality. Generating realistic, correlated uncertainty scenarios is important in many domains. In this paper, we use power systems as our motivating example, where the growing penetration of renewable energy makes it necessary to model wind, solar, and load trajectories jointly for more reliable decision-making under uncertainty. Whenever abstract or technical concepts arise, we illustrate them using this context, which also serves as the basis for our numerical experiments.

\subsection{Diffusion-Based Scenario Generation}
Diffusion models have recently emerged as effective methods for learning complex, high-dimensional data distributions~\citep{ho2020denoising,DiffusionSurvey,DiffusionVisionSurvey}. At a high level, they operate by gradually adding noise to data and then learning how to reverse this process, generating new samples through a sequence of denoising steps. This incremental procedure decomposes a difficult high-dimensional generation task into many simpler subproblems, which tends to yield more stable and tractable training dynamics than adversarial or likelihood-based methods. These properties make diffusion models particularly suitable for scenario generation in SP, where both sample diversity and structural fidelity are essential. In this work, we adopt the Denoising Diffusion Probabilistic Model (DDPM) \citep{ho2020denoising}, whose forward process progressively corrupts data with Gaussian noise and whose learned reverse process synthesizes new samples by iteratively removing that noise.

Specifically, let \( \xi^0 \in \mathbb{R}^D \) denote a data sample, such as a multivariate time series representing wind, solar, and load. The forward diffusion process defines a Markov chain that gradually perturbs \( \xi^0 \) over \(K\) diffusion steps. At each diffusion step \(k = 1, \ldots, K\), the sample \(\xi^k\) is generated as:
\begin{equation}
q(\xi^k \mid \xi^{k-1}) 
= \mathcal{N}\!\left(\xi^k;\ \sqrt{\alpha_k}\, \xi^{k-1},\, (1-\alpha_k)I\right),
\end{equation}
where \(\alpha_k \in (0,1)\) controls the signal-to-noise ratio at each diffusion step.

By marginalizing over previous steps, one obtains the closed-form distribution of \( \xi^t \) given the original sample \( \xi^0 \):
\begin{equation}
q(\xi^k \mid \xi^0) = \mathcal{N}(\xi^k; \sqrt{\bar{\alpha}_k}  \xi^0,\ (1 - \bar{\alpha}_k) I),
\end{equation}
where \( \bar{\alpha}_k = \prod_{s=1}^k \alpha_s \). This implies the following sampling equation for each $k = 1, \ldots, K$:
\begin{equation}\label{eq:Diff_Sample}
\xi^k = \sqrt{\bar{\alpha}_k}  \xi^0 + \sqrt{1 - \bar{\alpha}_k}  \epsilon, \quad \epsilon \sim \mathcal{N}(0, I), 
\end{equation}
which is used extensively during training to simulate noisy inputs.

The denoising process aims to probabilistically invert the forward corruption by learning, at each diffusion step \( k \), a conditional distribution of the previous state \( \xi^{k-1} \) given the current noisy state \( \xi^k \). This reverse transition could be formulated as
\begin{equation}
   p_\theta(\xi^{k-1} \mid \xi^k, c) \ \mathcal{N}\!\left(\xi^{k-1};\ \mu_\theta(\xi^k, c, k),\ \Sigma_\theta(\xi^k, c, k)\right), \label{eq:reverse} 
\end{equation}
where \( \mu_\theta(\cdot) \) and \( \Sigma_\theta(\cdot) \) are neural-network–parameterized functions with trainable parameters \( \theta \). The optional vector \( c \) represents exogenous contextual information used to condition the generation process. In our application, this may include observable system-level covariates such as time indices (such as hour of day, season) or weather-related variables. Conditioning allows the model to generate scenario samples that are statistically consistent with given external conditions rather than only following the unconditional historical distribution. Sampling from this reverse transition recursively for \( k = K, \ldots, 1 \), starting from a pure Gaussian noise vector \( \xi^K \sim \mathcal{N}(0, I) \), produces a synthetic scenario \( \hat{\xi}^0 \) that approximates a draw from the data-generating distribution conditioned on \( c \).

The DDPM framework generates scenarios by starting from Gaussian noise and iteratively denoising, producing samples that preserve key statistical and structural properties of the data, including temporal patterns and inter-variable dependencies. Compared with GANs, which rely on an adversarial min--max game and often suffer from unstable training and mode collapse (that is, poor coverage of data diversity), DDPMs offer more stable training and broader distributional coverage. In standard DDPM implementations, the learned reverse-transition functions \( \mu_\theta(\cdot) \) and \( \Sigma_\theta(\cdot) \) in \eqref{eq:reverse} are typically parameterized using U-Net \cite{ronneberger2015u} -- neural architectures that process data through local sliding windows and are particularly effective for capturing short-range patterns. This locality, however, makes such parameterizations less effective at modeling long-range temporal dependencies and complex cross-variable interactions common in high-dimensional multivariate time series. To overcome this limitation, we replace the CNN-based reverse model with a Transformer-based architecture in the denoising process to enhance the structural fidelity and decision relevance of the generated scenarios.

\subsection{Optimization-guided Training}
However, while diffusion models and Transformer-based architectures can capture complex, high-dimensional dependencies among uncertain variables, standard training methods focus solely on minimizing distributional divergence, without accounting for downstream decision performance.

In stochastic optimization problems, the value of generated scenarios lies not just in how realistic they appear, but in how informative they are for optimization. For instance, in long-term capacity planning, a generative model trained to reproduce average historical demand may systematically ignore rare peak events, leading to under-investment and degraded reliability. Thus, a model that appears statistically accurate may still yield poor decisions.

To bridge this gap, we introduce \textsc{Diff2SP} (Diffusion to Stochastic Programming), a framework that directly incorporates decision quality into the training of generative models. Rather than relying purely on statistical losses, \textsc{Diff2SP} aligns the training process with the ultimate goal of optimization. Specifically, we design a unified training objective that combines three components: 
\textit{(i)} a \emph{Noise Prediction Loss}, which follows standard DDPM training and guides the Transformer-based denoiser to accurately estimate the added noise at each diffusion step; 
\textit{(ii)} a \emph{Reconstruction Loss}, which encourages the model to recover temporally consistent samples from noise -- particularly important in multi-variate time series settings; and 
\textit{(iii)} an \emph{Optimization Loss}, which evaluates the impact of generated scenarios on the expected cost of the downstream stochastic decision problem.

Formally, at each diffusion timestep \( k = 1, \ldots, K \), the total loss is given by:
{\begin{equation}
\mathcal{L}^k(\theta) :=    
\lambda_\text{noise}  \mathcal{L}^k_\text{noise}(\theta) + \lambda_\text{recon}  \mathcal{L}^k_\text{recon}(\theta) + \lambda_\text{opt}  \mathcal{L}^k_\text{opt}(\theta), \label{eq:TotalLoss}
\end{equation}
}
where \( \lambda_\text{noise} \), \( \lambda_\text{recon} \), and \( \lambda_\text{opt} \) are hyperparameters that balance the tradeoff between fidelity to the data distribution and relevance to the downstream optimization problem. 

The Noise Prediction Loss ($\mathcal{L}^k_\text{noise}$) trains the model to reverse the forward diffusion process by estimating the Gaussian noise added to clean samples. Given a noisy input \( \xi^k \) at diffusion step \( k \), the model learns to predict the injected noise \( \epsilon \) using a neural network \( \epsilon_\theta(\xi^k, k) \). The loss is defined as, for each $k= 1, \ldots, K$, 
\begin{equation}
\label{eq:NoiseLoss}
    \mathcal{L}^k_\text{noise}(\theta) = \mathbb{E}_{\xi^0, \epsilon, k} \left[\| \epsilon - \epsilon_\theta(\xi^k, k) \|_2^2\right],\  
\end{equation}
where \( \epsilon \sim \mathcal{N}(0, I) \) is the known noise used to construct \( \xi^k \) as in \eqref{eq:Diff_Sample}. Minimizing this loss encourages accurate denoising and improves sample quality during generation.

The Reconstruction Loss ($\mathcal{L}^k_\text{recon}$) encourages the model to recover clean samples that closely match the original training data, promoting the preservation of temporal dynamics and cross-variable dependencies. It penalizes the discrepancy between the reconstructed sample and the ground-truth input. Specifically, for $k= 1, \ldots, K$, we define 
\begin{equation}\label{eq:recon_loss}
    \mathcal{L}^k_\text{recon}(\theta) = \mathbb{E}_{\xi^0, k} \left[\| \xi^0 - p_\theta(\xi^k, k) \|_2^2\right], 
\end{equation}
where \( p_\theta(\xi^k, k) \in \mathbb{R}^D \) denotes the model’s prediction of the original clean sample \( \xi^0 \) from its noisy version \( \xi^k \) at step \( k \).

The Optimization Loss ($\mathcal{L}^k_\text{opt}$) is unique to our framework, as it directly aligns scenario generation with the downstream optimization objective. It penalizes discrepancies in optimization outcomes when using generated versus historical data. Formally, we define:
\begin{equation}\label{eq:opt_loss}
\mathcal{L}^k_\text{opt}(\theta) = \mathbb{E}_{x, k}\left[\| F^*_{ \xi^0} - F^*_{p_\theta(\xi^k, k)} \|_2^2\right], \  k= 1, \ldots, K, 
\end{equation}
where \( \xi^0 \) is the ground-truth sample data, \( \xi^k \) is its noisy version at diffusion step \( k \), and \( p_\theta(\xi^k, k) \) is the model's reconstruction. The notation $F^*_{ \xi}$ is defined in \eqref{eq:SP} to denote the optimal value of an optimization problem under scenario \( \xi \). By minimizing this loss, the model learns to produce scenarios that yield similar decision outcomes as those based on real data, thereby capturing decision-critical structure rather than merely replicating statistical features.

Unlike standard loss terms, the optimization loss depends implicitly on tranining samples $\xi$. 
To enable gradient-based training, we embed the optimization problem as a differentiable layer within the neural network architecture, which requires computing gradients with respect to the diffusion model parameters \( \theta \) through the optimization layer. Specifically, for a given diffusion step $k = 1, \ldots, K$, we have 
\begin{equation}
\frac{\partial \mathcal{L}^k_{\text{opt}}}{\partial \theta}
= \mathbb{E}_{k} \Big[
  2 \big( F^{*}_{p_{\theta}(\xi^{k}, k)} - F^{*}_{\xi^{0}} \big)\left.
  \frac{\partial F^{*}_{\xi}}{\partial \xi}
  \right|_{\xi = p_{\theta}(\xi^{k}, k)} \times \frac{\partial p_{\theta}(\xi^{k}, k)}{\partial \theta}
\Big].
\label{eq_implicitBP}
\end{equation}
\textcolor{black}{
\begin{theorem}[Theorem 4.1, Fiacco and Ishizuka (1990)]
For the parametric optimization problem $P_0(\xi)$, suppose that the parameter set is $\Xi=\{\xi \in \mathbb{R}^d\}$, and that the functions $f$ and $\nabla_x f$ are continuous on $X \times \mathcal{N}(\xi)$, where $\mathcal{N}(\xi)$ is a neighborhood of $\xi \in \Xi$. Then the value function $F^*(\xi) := \min_{x\in X} f(x,\xi)$ is locally Lipschitz near $\xi$, directionally differentiable at $\xi$, and its directional derivative in direction $v \in \mathbb{R}^d$ is given by:
\begin{equation}
    D F^*(\xi; v) = \min_{x\in S(\xi)} \nabla_{\xi} f(x,\xi)^\top v,
\end{equation}
where $S(\xi)$ is the set of minimizers of $f(\cdot,\xi)$ over $X$.
\end{theorem}
Here the directional derivative of the value function can be expressed via
$\nabla_{\xi} f(x,\xi)$ evaluated at minimizers $x\in S(\xi)$.} With the differentiability of the optimal value function, we can write out the  derivative of the overall loss function at $k = 1, \ldots, K$:
\begin{equation}
\label{eq:chain_rule}
\frac{\partial \mathcal{L}^k}{\partial \theta}
= \lambda_{\text{noise}} \frac{\partial \mathcal{L}^k_{\text{noise}}}{\partial \theta}
+ \lambda_{\text{recon}} \frac{\partial \mathcal{L}^k_{\text{recon}}}{\partial \theta}
+ \lambda_{\text{opt}} \frac{\partial \mathcal{L}^k_{\text{opt}}}{\partial \theta}.
\end{equation}

\begin{remark}[Differentiability]\label{rmk:Diff}
We justify that the loss components $\mathcal{L}_{\text{noise}}$ and $\mathcal{L}_{\text{recon}}$ are differentiable with respect to the model parameters $\theta$. 
In the definition of the noise prediction loss in \eqref{eq:NoiseLoss},  $\epsilon_\theta$ is a neural network parameterized by $\theta$, and $\xi^k$ is a noisy version of the clean sample $\xi^0$ constructed via a forward diffusion process. For any fixed realization $(\xi^k, k, \epsilon)$, the integrand is a smooth function of $\theta$ provided that the neural network $\epsilon_\theta$ is composed of differentiable functions, such as affine transformations and activation functions (such as ReLU, GELU, or tanh), which are commonly used in practice. Assuming the network architecture is fixed and the domain of $\theta$ is open, the function $\epsilon_\theta(\xi^k, k)$ is differentiable with respect to $\theta$, and the gradient of the loss can be written as, for $k = 1, \ldots, K$,
{\begin{equation}
\frac{\partial \mathcal{L}^k_{\text{noise}}}{\partial \theta}
= \mathbb{E}_{\xi^0, \epsilon, k}
\left[-2(\epsilon-\epsilon_\theta(\xi^k,k))^\top
\frac{\partial \epsilon_\theta(\xi^k,k)}{\partial \theta}\right].
\label{eq:Diff_noise}
\end{equation}}
The interchange of differentiation and expectation is justified by the dominated convergence theorem, under the assumption that the network outputs and their derivatives grow at most polynomially and are integrable under the joint distribution of $(\xi^0, \epsilon, k)$, which is satisfied when the diffusion noise $\epsilon$ is Gaussian and $\epsilon_\theta$ has bounded weights or regularization.

For the reconstruction loss, $p_\theta$ in \eqref{eq:recon_loss} is another neural network parameterized by $\theta$, trained to reconstruct the clean input $\xi^0$ from the noisy sample $\xi^k$. Analogously, for any fixed $(\xi^0, \xi^k, k)$, the squared error is differentiable in $\theta$ due to the smoothness of $p_\theta$. Therefore, the gradient of the reconstruction loss is given by
{\begin{equation}
\frac{\partial \mathcal{L}^k_{\text{recon}}}{\partial \theta} \\
=  \mathbb{E}_{\xi^0, k} \left[ -2 \left(\xi^0 - p_\theta(\xi^k, k)\right)^\top \frac{\partial p_\theta(\xi^k, k)}{\partial \theta} \right]. \label{eq:Diff_recon}
\end{equation}
}
Again, the interchange of expectation and differentiation is justified by the dominated convergence theorem, assuming the outputs and gradients of $p_\theta$ are integrable under the distribution of $(\xi^0, t)$.
\end{remark}

\begin{remark}[Incompatibility of GANs with Optimization-Aware Training]
\label{remark_GAN}
While our approach focuses on diffusion models, it is natural to ask whether alternative generative frameworks, especially GANs, could also support optimization-aware training. GANs have been widely adopted for data generation tasks and may initially appear to be a viable alternative. However, their training paradigm presents several fundamental limitations that make them incompatible with decision-aligned objectives in SP. Standard GANs are trained using a minimax objective:
\begin{equation}
    \min_G \max_D ~ \mathbb{E}_{x \sim\mathcal{P}_{\text{data}}} [\log D(x)] 
    + \mathbb{E}_{z \sim \mathcal{N}(0, I)} [\log (1 - D(G(z)))],
\end{equation}
where \( z \sim \mathcal{N}(0, I) \) is a latent noise vector, \( G \) is the generator, and \( D \) is the discriminator. The generator learns to produce samples that maximize the discriminator’s classification error, without consideration of external task-specific objectives. Consequently, loss functions such as \( \mathcal{L}_{\text{opt}} \), which are unrelated to the discriminator’s feedback, cannot contribute to the generator’s gradient updates.

In contrast, Diff2SP is designed to support end-to-end training with optimization-aware loss functions. It relies on differentiable optimization layers, where gradients propagate through the optimization solution via implicit differentiation (see Eq.~\eqref{eq_implicitBP}) with \( \xi = G(z) \). This mechanism assumes that the generator \( G \) is smooth and stable. However, GANs are known to suffer from training instabilities such as mode collapse and vanishing gradients, particularly in data-scarce settings, rendering \( \partial \xi / \partial \theta \) unreliable or undefined in practice.

Furthermore, GANs do not provide explicit likelihoods or score functions \( \nabla \log p(\xi) \), which are essential for principled alignment between generated distributions and downstream optimization objectives. They also tend to underrepresent low-probability but high-impact scenarios, which are often critical in SP contexts. 
By contrast, Diff2SP avoids these limitations through its iterative denoising structure and attention-based modeling, enabling stable gradient flow and effective coverage of complex, multimodal uncertainties. In particular, the denoising network \( \epsilon_\theta(\xi^k, k) \) learned during training implicitly estimates the score function at each diffusion step, providing access to \( \nabla \log p(\xi^k) \). This score information can be used to guide distributional alignment or regularization strategies that directly account for both probabilistic fidelity and decision quality, which is incapable of adversarial frameworks like GANs.

\end{remark}

\section{Sample Efficiency of \textsc{Diff2SP}}\label{sec:Complexity}
In this section, we analyze the number of samples required for \textsc{Diff2SP} to achieve a target level of decision quality and compare this requirement against GAN. The key idea is to formalize regret, which measures the loss in decision quality due to learning from an approximate distribution, and apply a Probably Approximately Correct (PAC)–style analysis to link distributional accuracy to decision performance.

While our theoretical analysis of sample complexity does not depend on any particular generative model architecture, it rests on two essential properties: the scenario generator is trained with a loss function that explicitly accounts for downstream decision quality, such as in \eqref{eq:TotalLoss}, and the associated function class has bounded complexity, in the sense that its Rademacher complexity vanishes at rate \( O(1/\sqrt{N}) \) as the sample size \( N \) grows. 
Note that in principle, diffusion is not necessary to achieve the theoretical regret bounds; other generators, such as GANs or autoregressive models, could obtain comparable statistical guarantees if trained with an optimization-aware objective and sufficient regularization.  
In practice, however, diffusion models offer a distinct advantage because their denoising score-matching objective is a stable, fully differentiable learning problem. This property allows decision-dependent losses, often nonconvex and high-variance, to be incorporated seamlessly into end-to-end training through automatic differentiation. By contrast, adversarial models such as GANs involve two-player dynamics and often yield non-differentiable generator outputs, making it difficult to propagate decision-sensitive gradients effectively.  
Consequently, diffusion models combine theoretical soundness with practical tractability, providing a natural framework for decision-centric scenario generation.

We now formalize this setting to analyze the relationship between the learned and true distributions in terms of their induced decisions and resulting regret. Let $\mathcal{P}$ denote the true uncertainty distribution and $\mathcal{Q}_\theta$ the distribution produced by a generative model with parameters~$\theta$.
Let $ \tilde{x}_\theta \in \arg\min_{x \in \mathcal{X}} \mathbb{E}_{\xi \sim \mathcal{Q}_\theta}[f( x, \xi)] $ denote the decision optimized under $ \mathcal{Q}_\theta $, and let $ x^* \in \arg\min_{x \in \mathcal{X}} \mathbb{E}_{\xi \sim\mathcal{P}}[f( x, \xi)] $ denote the decision optimized under the true distribution $\mathcal{P} $, assuming both optimal solution sets are nonempty.
Then we can define the \emph{expected regret} as follows: 
	\begin{equation}\label{eq:regret}
		R(\theta)=\mathbb{E}_{\xi\sim\mathcal{P}}\bigl[f(\tilde{x}_\theta,\xi)\bigr]-
		\mathbb{E}_{\xi\sim\mathcal{P}}\bigl[f( x^*,\xi)\bigr].
	\end{equation}
    
To support the analysis in this section, we impose the following three standing assumptions.

\begin{assumption}[Lipschitz continuity in uncertainty]\label{as:Lip}
The uncertainty space $ \Xi \subset \mathbb{R}^m $ is compact, and the decision space $ \mathcal{X} \subset \mathbb{R}^n $ is also compact with a diameter $ D := \sup_{x, x' \in \mathcal{X}} \|x - x'\| $.  
The cost function $ f : \mathcal{X} \times \Xi \to \mathbb{R} $ is $ L $-Lipschitz continuous with respect to $ \xi $, uniformly over $ x \in \mathcal{X} $; that is, $|f( x, \xi) - f( x, \xi')| \le L \|\xi - \xi'\|$ for all $x \in \mathcal{X},\ \xi, \xi' \in \Xi.
$
\end{assumption}

\begin{remark}
Since Assumption~1 requires the uncertainty space $\Xi \subset \mathbb{R}^m$ 
to be compact, both the true distribution $\mathcal{P}$ and the generated distribution $Q_{\theta}$ are supported on a compact set. This immediately implies that $\mathcal{P}$ and $Q_{\theta}$ have finite first moments. Consequently, the 1-Wasserstein distance $W_1(\mathcal{P},Q_{\theta})$ between $\mathcal{P} $ and  $ \mathcal{Q}_\theta $, defined as
$
W_1(\mathcal{P}, \mathcal{Q}_\theta) := \inf_{\gamma \in \Pi(P, \mathcal{Q}_\theta)} \mathbb{E}_{( \xi, \tilde{\xi}) \sim \gamma} \left[ \|\xi - \tilde{\xi}\| \right]$,
with $ \Pi(\mathcal{P}, \mathcal{Q}_\theta) $ being the set of all joint distributions (couplings) with marginals $\mathcal{P} $ and $ \mathcal{Q}_\theta $, is well-defined. 
\end{remark}


\begin{assumption}[Bounded capacity via Rademacher complexity]\label{as:Rademacher}
Let $\mathcal{F}=\{f(x_\theta,\cdot):\theta\in\Theta\}$ denote the function class
induced by the generative model through its decisions $x_\theta$.
For a sample $S=(\xi_1,\dots,\xi_N)\sim\mathcal{P}^N$ and independent Rademacher random
variables $\sigma_1,\dots,\sigma_N\sim\mathrm{Unif}\{\pm 1\}$, the empirical
Rademacher complexity of $\mathcal{F}$ is defined as
\[
\mathcal{R}_N(\mathcal{F}):= \mathbb{E}_{S,\sigma}\left[
   \sup_{f\in\mathcal{F}} \frac{1}{N}\sum_{i=1}^N \sigma_i f(\xi_i)
\right].
\]

We assume that $\mathcal{F}$ has bounded capacity in the sense that $\mathcal{R}_N(\mathcal{F}) \le \frac{C}{\sqrt{N}}$, for some constant $C>0$ independent of $N$. 
\end{assumption}

\begin{remark}
When $\mathcal{F}$ is induced by a deep neural network $\epsilon_\theta$ with $L_{\text{net}}$ layers and 1-Lipschitz activations such as ReLU, \citet{bartlett2017spectrally,neyshabur2017role} show that
\[
\mathcal{R}_N(\mathcal{F})
\le \frac{C}{\sqrt{N}} \cdot \prod_{l=1}^{L_{\text{net}}}\|\mathbf{W}^{(l)}\|_2.
\]
With spectral normalization, $\ell_2$ weight decay, or the implicit bias of SGD,
it is standard to enforce $\prod_{l=1}^{L_{\text{net}}}\|\mathbf{W}^{(l)}\|_2 \le \mathbf{B}$ 
for a constant $\mathbf{B}$ independent of $N$.  
This yields $\mathcal{R}_N(\mathcal{F}) = O(1/\sqrt{N})$, as required for the sample
complexity analysis.
\end{remark}

Before presenting the main theoretical results, we establish a relationship
between the accuracy of the generated distribution and the quality of the resulting decision.
Under the Lipschitz continuity assumption, the decision regret can be directly bounded
by the 1-Wasserstein distance between the true and generated distributions.
This result serves as the analytical link between distributional learning and
decision performance and supports all subsequent complexity analyses.
The proof of Lemma~\ref{lem:RegretBd} is provided in Appendix~\ref{sec:App_RegretBd}.

\begin{lemma}[Regret bound]\label{lem:RegretBd}
	Under Assumption \ref{as:Lip}, we have that 
	\begin{equation}\label{eq:RegretBound}
	R(\theta) \le 2L \cdot W_1(\mathcal{P}, \mathcal{Q}_\theta).  
	\end{equation}
\end{lemma}

By the triangle inequality of the 1-Wasserstein distance \citep[p.94]{OptimalTransport}, for any generator $\mathcal{Q}_\theta$ supported on $\Xi$,
\begin{equation}
W_1(\mathcal{P}, \mathcal{Q}_\theta) \le W_1(\mathcal{P}, \mathcal{P}_N) + W_1(\mathcal{P}_N, \mathcal{Q}_\theta),
\end{equation}
where $W_1(\mathcal{P}, P_N)$ is the sampling error and $W_1(\mathcal{P}_N, \mathcal{Q}_\theta)$ is the generating error. Since the sampling term is method-independent and converges at a known rate, we focus on the generating error to characterize and compare the sample complexities of Diff2SP and GAN.

For any decision $x\in\mathcal{X}$, define the true and empirical objectives $L(x):=\mathbb{E}_{\xi\sim\mathcal{P}}[f(x,\xi)]$ and $\widehat L_N(x):=\frac1N\sum_{i=1}^N f(x,\xi_i)$. Given a generator $\mathcal{Q}_\theta$, the corresponding decision is obtained by solving $x_\theta\in\arg\min_{x\in\mathcal{X}} \mathbb{E}_{\xi\sim \mathcal{Q}_\theta}[f(x,\xi)]$, while the benchmark decision under the true distribution is $x^*\in\arg\min_{x\in\mathcal{X}} L(x)$.
Let $\hat\theta\in\arg\min_{\theta\in\Theta}\widehat L_N(x_\theta)$ be the learned parameter and $\theta^\star\in\arg\min_{\theta\in\Theta} L(x_\theta)$ the best-in-class parameter. The performance of the learned model is assessed through the regret
\begin{equation}\label{eq:regret}
R(\hat{\theta}) := L(x_{\hat{\theta}}) - L(x^*),
\end{equation}
which measures the gap between the expected cost of the decision induced by the learned parameter $\hat{\theta}$ and that of the true optimal decision $x^*$. In other words, regret captures how much decision quality is lost due to both statistical error in estimating $\mathcal{P}$ and approximation error in representing $\mathcal{P}$ by the model class $\{\mathcal{Q}_\theta\}_{\theta\in\Theta}$.

\begin{theorem}[Sample complexity for decision quality of Diff2SP]
\label{thm:diff2sp-bound}
Fix $\varepsilon,\delta \in (0,1)$. Suppose Assumptions~\ref{as:Lip} and \ref{as:Rademacher} hold, and the empirical optimization error satisfies $\epsilon_{\mathrm{opt}} \le \varepsilon/(4L)$. Let $\theta^\star \in \arg\min_{\theta\in\Theta} L(\theta)$ be the best-in-class parameter. Then for any
\[
N \geq C \frac{C_{\mathrm{nn}}^2 + (LB)^2 \log(1/\delta)}{\varepsilon^2},
\]
where $C>0$ an absolute constant, we have with probability at least $1-\delta$,
\[
R(\hat\theta) \le \varepsilon + \big(L(\theta^\star)-L(x^*)\big).
\]
\end{theorem}
{
\textit{Proof Sketch.} The result follows by decomposing the total error into the statistical estimation error and the optimization-induced error. We first bound each component using the regularity conditions and the concentration properties established earlier. 
Combining these bounds yields the stated rate, with the remaining technical details deferred to Appendix~\ref{app:thm2}.}
\begin{remark}[On realizability of diffusion models]
Regarding the last term $\big(L(\theta^\star)-L(x^*)\big)$, in practice the data distribution $\mathcal{P}$ need not lie exactly in the parametric family $\{\mathcal{Q}_\theta:\theta\in\Theta\}$. At the population level, however, exact realizability is mathematically well established: if one has access to the true score function $\nabla_\xi \log p_t(\xi)$ of the forward noising process, then simulating the corresponding reverse-time SDE yields samples exactly from the data distribution $\mathcal{P}$ \citep{anderson1982reverse,song2021score}. By Theorem~2.1 of \citet{anderson1982reverse} and Theorem~2.1 of \citet{song2021score}, the time-reversal of a diffusion process is again an SDE whose drift term involves the score $\nabla_\xi \log p_t(\xi)$. In particular, if the true score function is known, then simulating the reverse-time SDE yields samples exactly from the original data distribution $\mathcal{P}$.

In this sense, diffusion models are realizable in the nonparametric infinite-data setting. When the score function is approximated by a neural network, universal approximation theorems guarantee that sufficiently large networks can approximate the true score arbitrarily well on compact domains \citep{cybenko1989approximation,hornik1991approximation}. Thus, while exact realizability cannot be ensured in finite samples, it is reasonable to interpret the realizability assumption as an idealized case in which the approximation error is negligible.
\end{remark}

\begin{theorem}[Sample complexity for decision quality of GAN]
\label{thm:gan-bound}
Assume Assumptions 1–2 hold, and consider a vanilla JS–GAN trained with a discriminator class of VC-dimension $p$ whose discriminators are $L_p$-Lipschitz with respect to their parameters. Suppose the measuring function is $L$-Lipschitz and bounded in $[-\Delta,\Delta]$. Assume the generator approximately minimizes the empirical Jensen--Shannon (JS) divergence with optimization residual $\epsilon_{\mathrm{opt}} \le \varepsilon^2 / (16 L^2 B^2)$, and that the generator class satisfies $\inf_{\theta\in\Theta}\mathrm{JS}(\hat{\mathcal{P}},\mathcal{Q}_\theta)\le \varepsilon_0$. Then, for any target accuracy $\varepsilon>0$ and confidence level $\delta\in(0,1)$, to guarantee $\Pr\!\bigl[R(\hat\theta)\le\varepsilon\bigr] \ge 1-\delta$, it suffices that 
\[
N \geq\frac{C_1^{2}\,p^{2}
      \log\Big(\tfrac{L\,L_p}{\varepsilon\delta}\Big)}
     {\Big(\tfrac{\varepsilon^{2}}{16L^{2}B^{2}}-\varepsilon_0\Big)^{2}},
\]
where $C_1>0$ is a universal constant.
\end{theorem}
{\textit{Proof Sketch.}
The proof builds on Theorem \ref{thm:diff2sp-bound} and shows that the proposed procedure preserves the relevant asymptotic property under the stated assumptions. The key step is to control the perturbation introduced by the learned approximation and show that it is asymptotically negligible. 
Substituting this bound into the target objective establishes the desired conclusion. See Appendix~\ref{app:thm3} for the full argument.}

Theorems~\ref{thm:diff2sp-bound} and \ref{thm:gan-bound} are directly comparable since both quantify decision regret. For JS--GANs, controlling the JS divergence between $\mathcal{P}$ and $Q_{\hat\theta}$ yields a decision regret rate of $\mathcal{O}(N^{-1/4})$. In contrast, Diff2SP minimizes an optimization-aware loss that directly bounds the decision risk, yielding the sharper $\mathcal{O}(N^{-1/2})$ rate. 
Thus, while both approaches provide regret guarantees, GANs achieve them indirectly through statistical divergence control, whereas Diff2SP provides a direct guarantee aligned with the downstream decision objective. 
This distinction explains the stronger theoretical rate and the improved empirical robustness of Diff2SP, particularly its ability to avoid mode-collapse behavior that can impose a constant regret floor in GANs. The difference is further captured by the approximation term $\varepsilon_0 := \inf_{\theta\in\Theta} \mathrm{JS}(\mathcal{P},\mathcal{Q}_\theta)$ in Theorem~\ref{thm:gan-bound}. A non-negligible $\varepsilon_0$ can lead to an irreducible component in the regret bound for GAN-based generators. By contrast, Diff2SP directly incorporates decision-aware training through a differentiable denoising objective, making the generated scenarios more closely aligned with downstream SP performance.

\begin{remark}
Our training objective is a weighted combination of the Noise Prediction Loss, the Reconstruction Loss, and the Optimization Loss (i.e., regret). The theoretical regret bounds established in Section~\ref{sec:Complexity} rely on the assumption that the Optimization Loss is explicitly minimized. While our training jointly optimizes all three losses, the regret analysis remains valid under the condition that the Optimization Loss carries sufficient weight in the overall objective. In practice, we ensure that this component is prioritized by choosing a non-negligible weight $ \lambda_3 $, so that minimizing the total objective also yields low regret.
\end{remark}

\section{Experiments}
\label{sec_experiment}
We evaluate Diff2SP through a set of experiments that assess its ability to generate scenarios supporting high quality decisions in SP. Our evaluation considers both statistical fidelity and downstream decision performance under uncertainty. We conduct two experiments, one based on synthetic data and the other on real world data. The synthetic experiments provide a controlled setting for examining structural fidelity and robustness, while the real world case study illustrates the applicability of Diff2SP to complex multivariate time series arising in power systems. Across all experiments, we compare Diff2SP against GAN and study performance under varying levels of data availability and optimization constraints.

\subsection{Numerical simulation}
We begin with a controlled numerical experiment to assess whether Diff2SP can recover structured dependencies embedded in uncertainty distributions. Specifically, we construct a synthetic multivariate Gaussian distribution $\mathcal{N}(0, \Sigma)$, where the covariance matrix $\Sigma \in \mathbb{R}^{d \times d}$ is manually designed to encode a prescribed dependency structure.
Unlike real-world datasets, where such structural dependencies are implicit or noisy, this synthetic design provides a clean testbed for validating statistical fidelity. We set $d = 20$ and generate 50,000 training samples from $\mathcal{N}(0, \Sigma)$. The diffusion model is trained for $200$ epochs with a batch size of $64$ and a learning rate of $10^{-2}$ using the Adam optimizer \cite{kingma2014adam}. To ensure that the learned correlations arise solely from the generative model, we train the diffusion model without any contextual or conditioning inputs.

After training, we draw $n$ samples from the diffusion model and assess \emph{structural fidelity} by comparing the empirical correlation matrix $\widehat{C}(n)$ computed from generated samples with the ground-truth correlation structure $\Sigma$. We quantify the discrepancy using the off-diagonal root mean squared error (RMSE); specifically, $\mathrm{RMSE}(\widehat{C}, \Sigma)
:= \sqrt{
\frac{2}{d(d-1)}
\sum_{1 \le i < j \le d}
\left( \widehat{C}_{ij} - \Sigma_{ij} \right)^2
}$.
Table~\ref{tab:distance} reports the off-diagonal RMSE under different numbers of generated samples \( n \). For both GAN and Diff2SP, the RMSE generally decreases as \( n \) increases, indicating improved recovery of the underlying dependency structure with more samples. Diff2SP consistently attains lower RMSE than GAN in the low-sample regime, demonstrating superior sample efficiency in capturing correlation structure. As \( n \) grows, the performance gap narrows, reflecting convergence of both methods. These empirical trends are consistent with our theoretical sample complexity results, which predict faster convergence for Diff2SP.

\begin{table}[h]
\centering
\caption{Distance comparison under different sample sizes}
\label{tab:distance}
\begin{tabular}{c *{7}{c}}
\toprule
$n$ & 3 & 5 & 10 & 30 & 60 & 100 & 200 \\
\midrule
GAN & 0.8233 & 0.4945 & 0.2051 & 0.2031 & 0.1343 & 0.1303 & 0.1182 \\
Diff2SP & 0.6397 & 0.2885 & 0.1819 & 0.1486 & 0.1106 & 0.1226 & 0.1059 \\
\bottomrule
\end{tabular}
\label{tab:distance}
\end{table}
Figure~\ref{fig_corr} visualizes correlation recovery by comparing the ground-truth matrix with correlation matrices estimated from Diff2SP-generated samples using $n=6$, $20$, and $100$ samples. With only $n=6$ samples, the estimate is noisy, but the global block structure and the dominant within-group dependencies are already visible. At $n=20$, spurious fluctuations are reduced and the main cross-block relationships begin to stabilize. By $n=100$, the estimated matrix closely matches the ground truth, with both block patterns and finer within-block variation largely recovered. Overall, these heatmaps show that Diff2SP learns the underlying joint dependency structure and that structural fidelity improves predictably as more generated samples are used to estimate correlations.

\begin{figure}[h]
  \centering
  \begin{subfigure}[t]{0.45\linewidth}
    \centering
    \includegraphics[width=\linewidth]{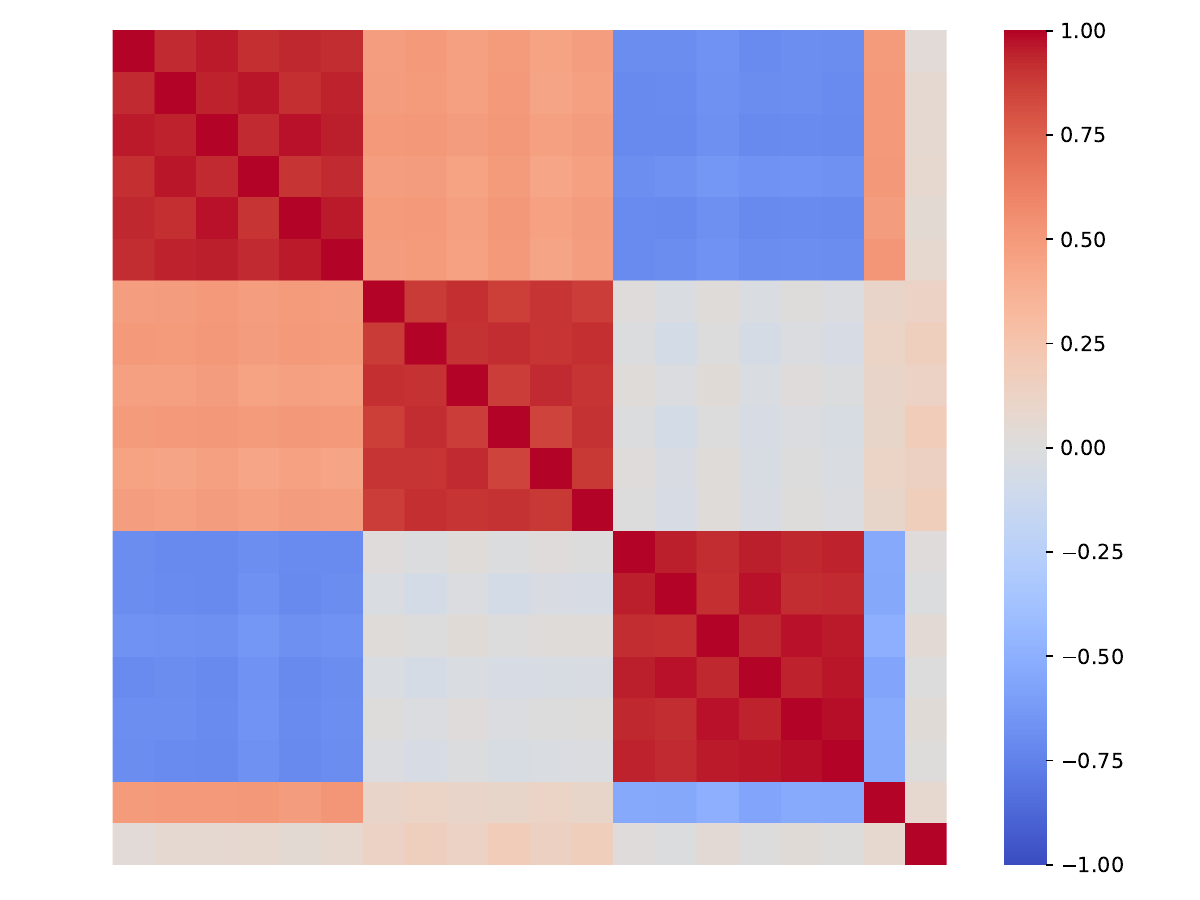}
    \caption{Real Correlation}
  \end{subfigure}%
  \hfill%
  \begin{subfigure}[t]{0.45\linewidth}
    \centering
    \includegraphics[width=\linewidth]{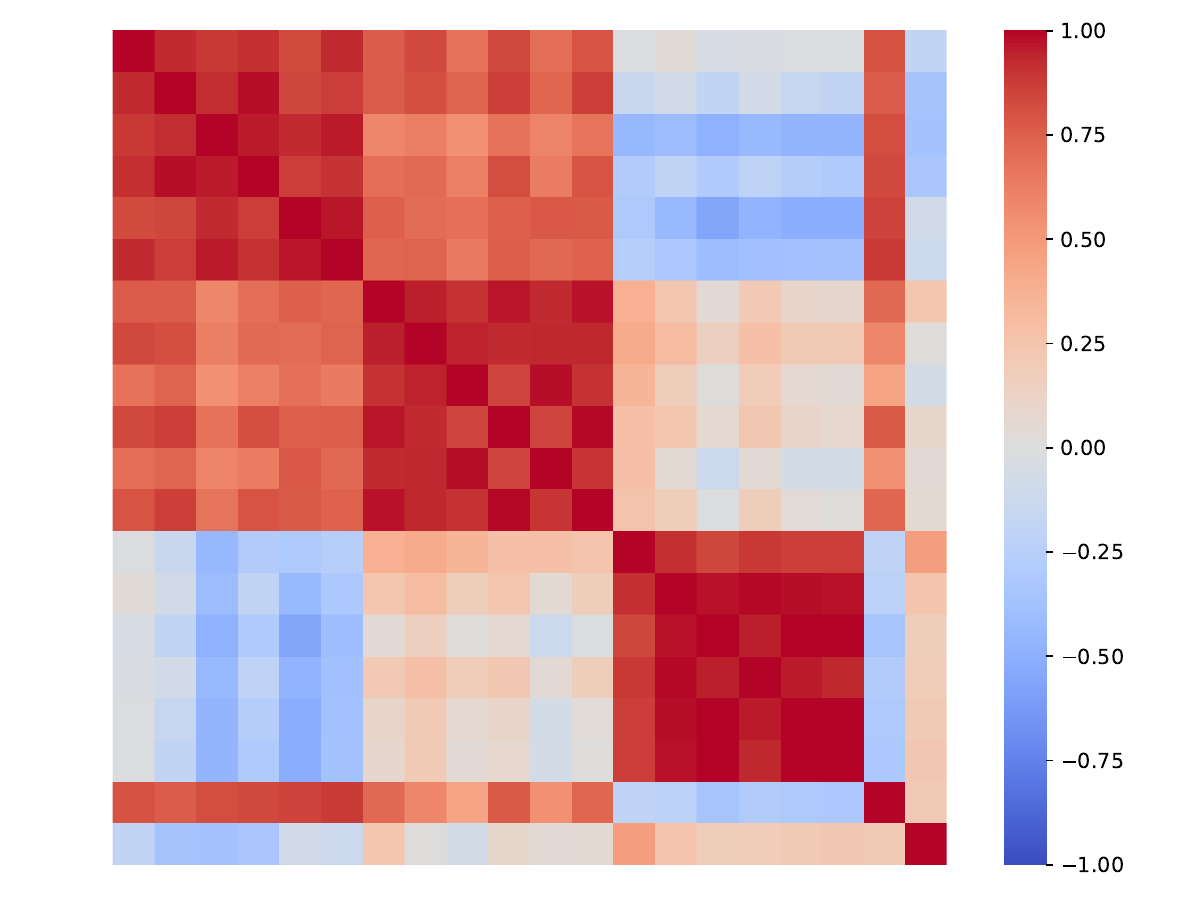}
    \caption{Sample size 6}
  \end{subfigure}

  \begin{subfigure}[t]{0.45\linewidth}
    \centering
    \includegraphics[width=\linewidth]{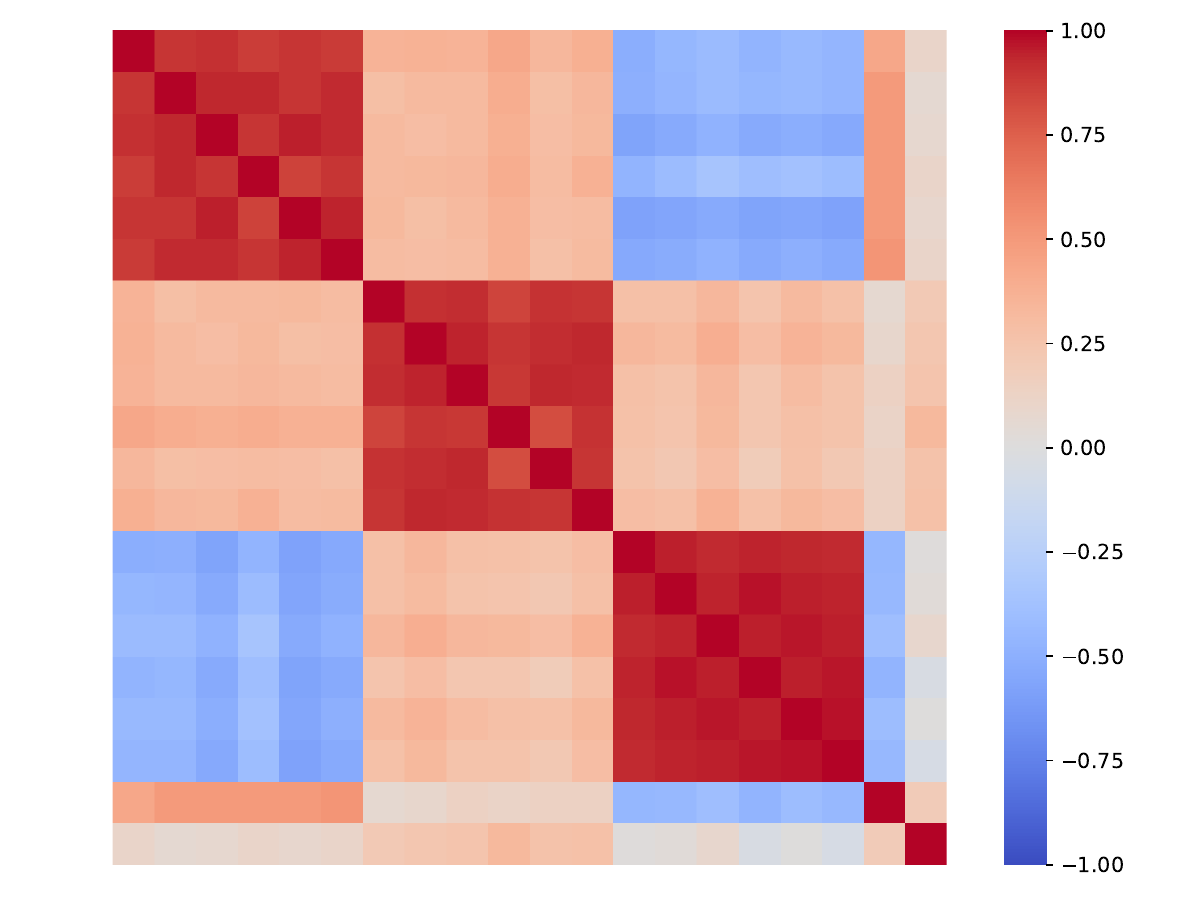}
    \caption{Sample size 100}
  \end{subfigure}%
  \hfill%
  \begin{subfigure}[t]{0.45\linewidth}
    \centering
    \includegraphics[width=\linewidth]{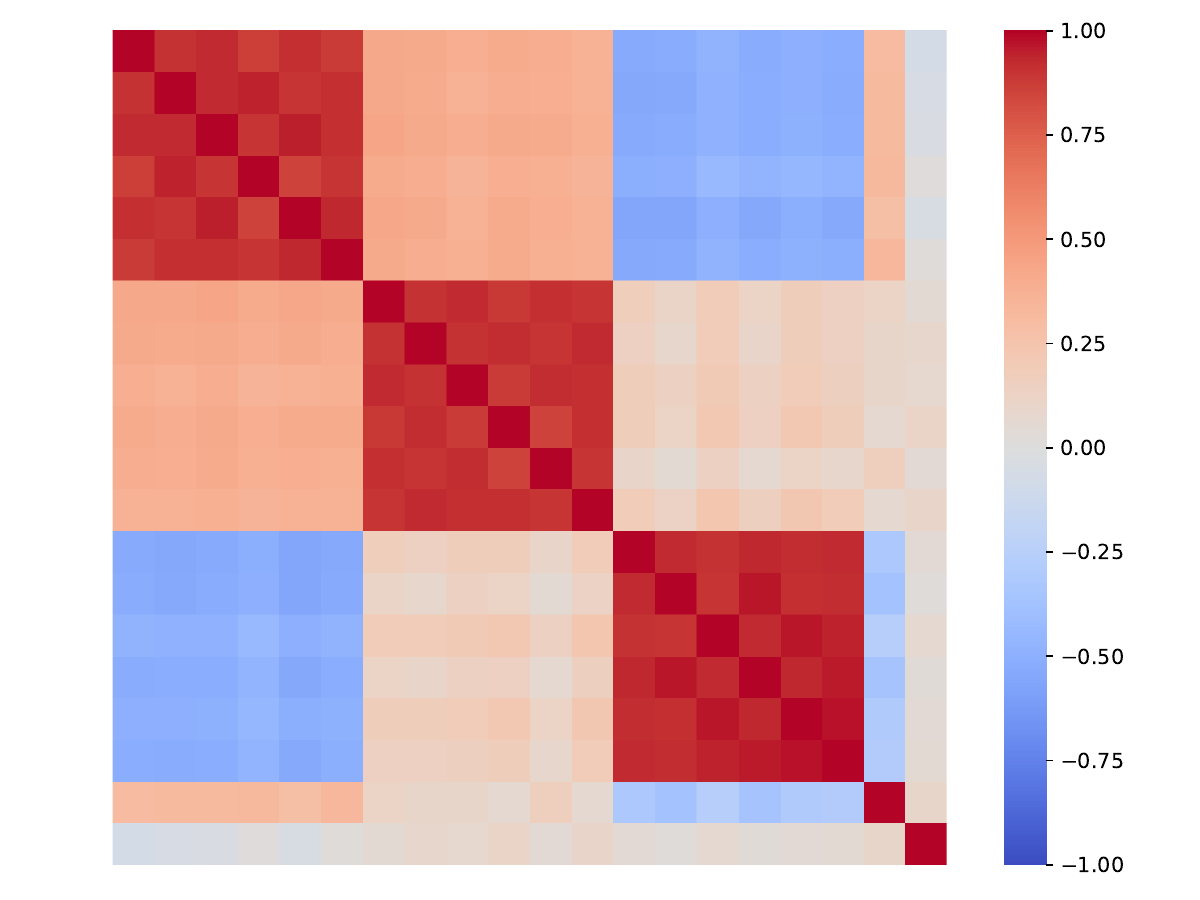}
    \caption{Sample size 500}
  \end{subfigure}
  \caption{Comparison between the ground‑truth correlation matrix and the estimated correlations under different sample sizes.}
  \label{fig_corr}
\end{figure}

Since Diff2SP is designed as a general framework that can be coupled with a wide range of stochastic optimization problems, we consider a quadratic program as a representative downstream task to evaluate the impact of Diff2SP-generated scenarios on decision quality. Specifically, we study the quadratic programming problem $Q( x^\star;{ \xi})=\{\min_{x} \frac{1}{2} x^\top P x + q(\xi)^\top x\mid A x \ge b; G x = h\}$, where $P \in \Re^{18 \times 18} \succeq 0$ and $(A,b,G,h)$ are fixed across experiments, $A \in \Re^{10 \times 18}$, $G \in \Re^{5 \times 18}$ and $q(\cdot):\mathbb{R}^{18} \rightarrow\mathbb{R}^{18}$ is a deterministic (linear) mapping $q(\xi) = W\xi$, with $W \in \mathbb R^{18\times20}$, that maps the 20-dimensional uncertainty vector $\xi$ to the 18 linear cost coefficients.
deterministic mapping from the uncertainty vector $\xi$ to the linear cost coefficients. In the synthetic study, we specify a ground-truth distribution $\mathcal{P}$ for $\xi$ and draw i.i.d.\ historical samples $\{\xi^{(i)}\}_{i=1}^N \sim \mathcal{P}$ for training. After training, Diff2SP induces a model distribution $\mathcal{Q}_\theta$, from which we generate scenarios $\{\hat{\xi}^{(j)}\}_{j=1}^n \sim \mathcal{Q}_\theta$. The QP is then instantiated using the generated coefficients $q(\hat{\xi}^{(j)})$. The optimization layer is implemented using the \texttt{qpth} package \cite{amos2017optnet}, which enables efficient forward and backward passes and allows gradients to propagate through the quadratic program during training.

To complement the correlation recovery analysis, we evaluate the quality of generated scenarios from a decision-making perspective using the optimization error
\begin{equation}
\label{eq_error}
    \text{err} = \Bigg\vert\frac{Q( x^\star;\hat{ \xi})-Q( x^\star;{ \xi})}{Q( x^\star;{ \xi})}\Bigg\vert,
\end{equation}
where $ \xi$ is the historical data and $\hat{ \xi}$ is the generated data. This metric quantifies the relative difference in expected cost between the solution obtained from generated scenarios and the true optimal solution under the ground-truth distribution. Unlike binary failure thresholds, this continuous measure captures both the stability and accuracy of scenario-based optimization outcomes.

\begin{wrapfigure}{r}{.5\textwidth} 
  \centering
  \vspace{-\baselineskip} 
  \includegraphics[width=\linewidth]{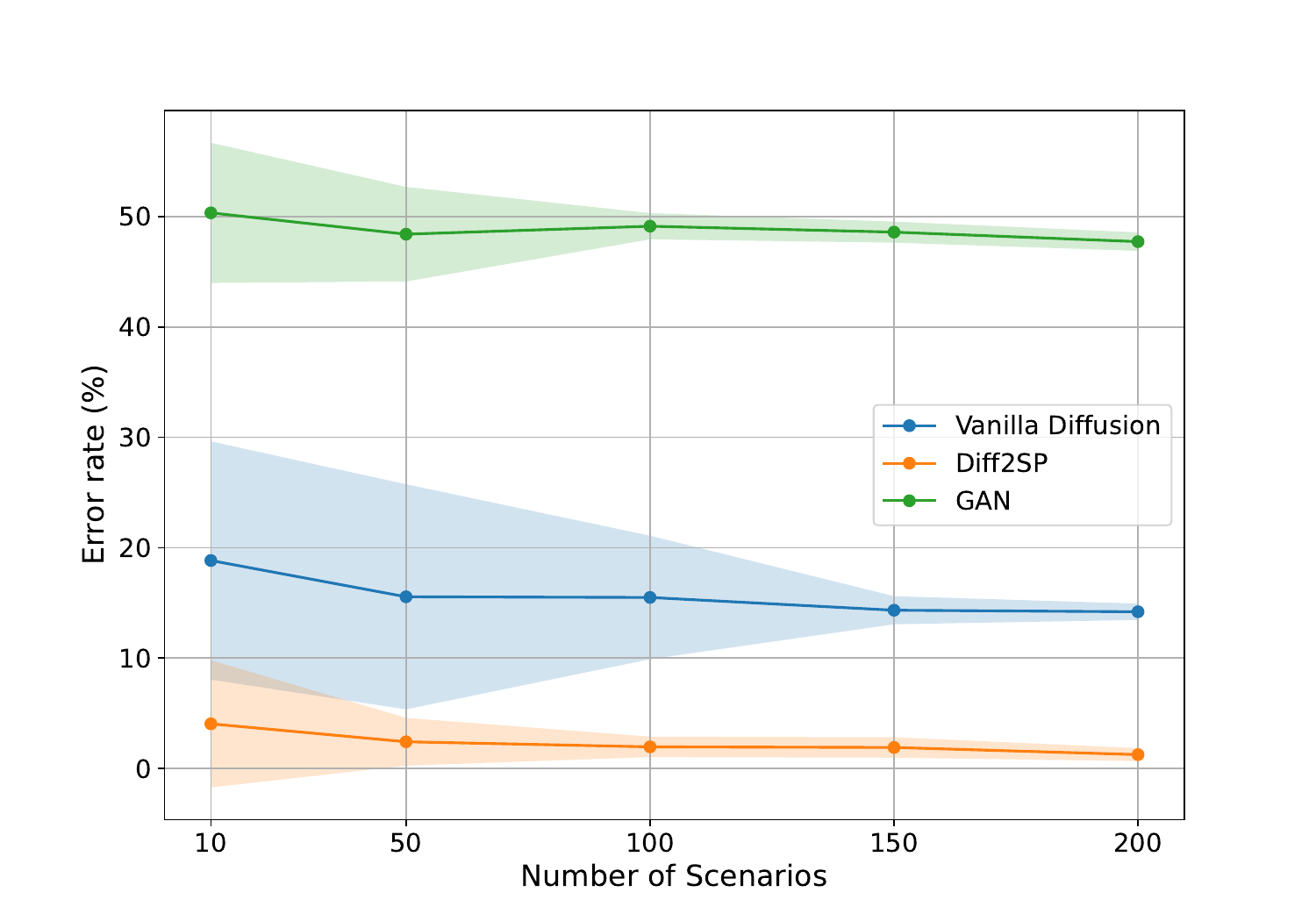}
\caption{Error rate vs Number of Scenarios}
    \label{fig:error}
\end{wrapfigure}
Figure~\ref{fig:error} reports the optimization error achieved by GAN, vanilla diffusion, and Diff2SP as a function of the number of generated scenarios. Diff2SP consistently attains lower optimization error than the baseline methods across all sample sizes, with the largest performance gains appearing in the low-sample regime. As the number of generated scenarios increases, the optimization error decreases for all methods, reflecting improved approximation of the underlying uncertainty distribution; however, substantial performance differences persist across models. In particular, GAN-based generation yields higher optimization error throughout, indicating weaker alignment with decision-relevant structure. This gap is consistent with known limitations of GAN training in multivariate stochastic settings, where the discriminator primarily evaluates marginal or low-order joint statistics and may fail to enforce higher-order dependency patterns that influence downstream optimization outcomes. As a result, GAN-generated scenarios provide a less accurate representation of the uncertainty structure required for reliable decision-making, whereas Diff2SP maintains a consistent advantage by explicitly incorporating decision-aware training objectives.

In addition, we observe that the variance of the optimization error across trials decreases as the number of generated scenarios increases. This reflects the stabilizing effect of larger scenario sets, which provide a more reliable approximation of the underlying uncertainty distribution and reduce variability in optimization outcomes induced by sampling noise. 

Overall, these numerical results demonstrate that Diff2SP achieves both high statistical fidelity and strong decision relevance, supporting its effectiveness for decision-making under uncertainty.

\subsection{Real-World Case Study}
To evaluate the practical applicability of Diff2SP, we apply the same experimental setup to a real-world SP task in power systems. Although the optimization problem structure and evaluation protocol remain the same as in the synthetic experiments, the uncertainty in this setting is derived from real-world multivariate time-series data. In such cases, the true distributional structure is unknown and may involve non-Gaussian characteristics or nonlinear dependencies, making the problem inherently more complex and less controllable than in the synthetic setup. 

\subsubsection{Data Setup}
\label{sec:datasetup}
{ We analyzed two years of historical data encompassing wind, solar, and load time series across six aggregated Local Resource Zones (LRZs) in the MISO footprint.\footnote{The original MISO LRZs are aggregated into six zones: LRZ1, LRZ2\_7, LRZ3\_5, LRZ4, LRZ6, and LRZ8\_9\_10. }, sampled at 5-minute intervals. So at each time slot, the uncertainty vector $\xi\in\mathbb{R}^{3\times6}$.} This dataset comprises a total of 210,240 data points. We used the electricity generation dataset by NERL \footnote{https://github.com/PERFORM-Forecasts/documentation.} for training, and the data was processed and categorized based on the following criteria for further analysis:

{\textbf{Time periods and ramping categories.}
To account for systematic intraday variation in renewable generation and load demand, we partition each day into four time periods. The early-morning period, from midnight to 8 a.m., corresponds to relatively low demand and generation levels. The morning period, from 8 a.m. to noon, captures the transition during which demand and renewable generation begin to increase. The afternoon period, from noon to 5 p.m., is associated with peak renewable generation, particularly solar generation. The evening period, from 5 p.m. to midnight, corresponds to a high-demand regime in which supply is more reliant on wind and other non-solar resources.

In addition to time-of-day effects, we characterize short-term temporal variability through ramping behavior. For each time step $t$, we compute the relative change between the current value and the value observed 30 minutes earlier as
$
r_t = \frac{{\bf \bar{z}}_t - {\bf \bar{z}}_{t-6}}{{\bf \bar{z}}_t},
$
where ${\bf \bar{z}}_t$ denotes the average generation or load at time step $t$, and ${\bf \bar{z}}_{t-6}$ denotes the corresponding value six 5-minute intervals earlier. A 30-minute lookback window was adopted following preliminary data exploration to characterize short-term ramping behavior while avoiding excessive sensitivity to fluctuations between adjacent 5-minute intervals. Based on the empirical distribution of the resulting ramping ratios, a threshold of 4\% was selected to distinguish significant ramping events from routine variations. Specifically, observations are classified into four regimes: significant positive ramping when $r_t \geq 4\%$, moderate positive ramping when $0 \leq r_t < 4\%$, moderate negative ramping when $-4\% < r_t \leq 0$, and significant negative ramping when $r_t \leq -4\%$. These categorical descriptors provide a compact representation of short-term temporal dynamics while preserving information about both the direction and magnitude of local changes in generation and load.
}

In this way, the historical data are categorized into 16 classes. This categorization provides a granular view of generation dynamics, enabling a clearer understanding of the stability and potential volatility in each time period. It also offers insights into operational adjustments needed to balance supply and demand more effectively.

\subsubsection{Results}
We trained Diff2SP on real-world power system data, including wind, solar, and load time series across multiple regions, with the goal of capturing the structural uncertainties inherent in renewable generation and electricity demand. Once training was complete, we generated 800 scenarios per class and solved IEEE-30 DC Optimal Power Flow (OPF) problem using the Progressive Hedging Algorithm (PHA).  Each sample is a multivariate time series and Diff2SP generates scenarios $\hat{\xi}\in\mathbb{R}^{L\times d}$, where $d=18$ corresponds to wind, solar, and load variables across 6 regions. Since the IEEE-30 test system is defined on 30 buses, we convert each generated regional scenario $\hat{\xi}^s$ into a bus-level demand realization $p_d^s\in\mathbb{R}^{30}$, which is regarded as the downstream uncertainty vector $\xi$ in the stochastic optimization layer. We then allocate regional loads to buses via a linear mapping $p_d^s = A\ell^s$ where $A\in\mathbb{R}^{30\times 18}$ is an allocation matrix. This mapping preserves the original 18-variable data representation used for training while producing physically meaningful bus-level demands required by DC-OPF.

Given a demand realization $p_d^s$, we solve the following convex DC-OPF. The decision variables are generator dispatch $p_g^s$ and bus voltage angles $\theta^s$. Let $\mathcal{N}$ denote the set of buses with $|\mathcal{N}|=30$, $\mathcal{E}$ the set of transmission lines, and $\mathcal{G}$ the set of generators. The per-scenario DC-OPF is
\begin{equation}\label{dcopf}
\begin{aligned}
F^*(p_d^s):=&\min_{p_g^s,\theta^s}\sum_{g\in\mathcal{G}}\Big(a_g (p_{g}^s)^2 + b_g p_{g}^s\Big)\\
\text{s.t.}~&\sum_{g\in\mathcal{G}(i)} p_{g}^s - p_{d,i}^s 
= \sum_{(i,j)\in\mathcal{E}} f_{ij}^s, ~\forall i\in\mathcal{N}\\
&f_{ij}^s ~=~ B_{ij}\,(\theta_i^s-\theta_j^s), ~ \forall (i,j)\in\mathcal{E},\\
&-\overline{F}_{ij}\le f_{ij}^s \le \overline{F}_{ij}, ~ \forall (i,j)\in\mathcal{E},\\
&\underline{P}_g \le p_{g}^s \le \overline{P}_g, ~ \forall g\in\mathcal{G},\\
&\theta_{r}^s = 0,\quad \forall i\in\mathcal{N}.
\end{aligned}
\end{equation}
Here $B_{ij}$ is the line susceptance, $\overline{F}_{ij}$ is the line flow limit, $r$ denotes the reference bus, the superscript $s$ indexes different scenarios. {The downstream OPF problem is evaluated on the standard IEEE 30-bus benchmark system, using the case30 network data provided by PYPOWER \cite{PYPOWER}.
}

To assess the quality of the generated data, we conducted two complementary evaluations. We first assess the statistical and temporal fidelity of the generated scenarios. We then adopt a comprehensive, class-wise evaluation that measures optimization- and network-relevant fidelity under downstream DC-OPF, capturing whether generated scenarios preserve the structures that matter for decision making and constraint behavior.

Figure \ref{fig:average} presents the comparison of original (blue) vs. generated (orange) data in 16 classes within a specific zone. The results show that the generated data effectively captures major peaks and valleys, demonstrating strong alignment with real-world patterns. However, minor discrepancies were observed in some high-generation regions, where the generated data occasionally overestimated peak values. {This tendency suggests that the model matches the overall distribution well, but may be slightly miscalibrated for extreme peaks. We therefore further assess whether such tail errors translate into meaningful degradation in downstream decision quality.}
\begin{figure*}[htbp]
    \centering
    \includegraphics[width=1\linewidth]{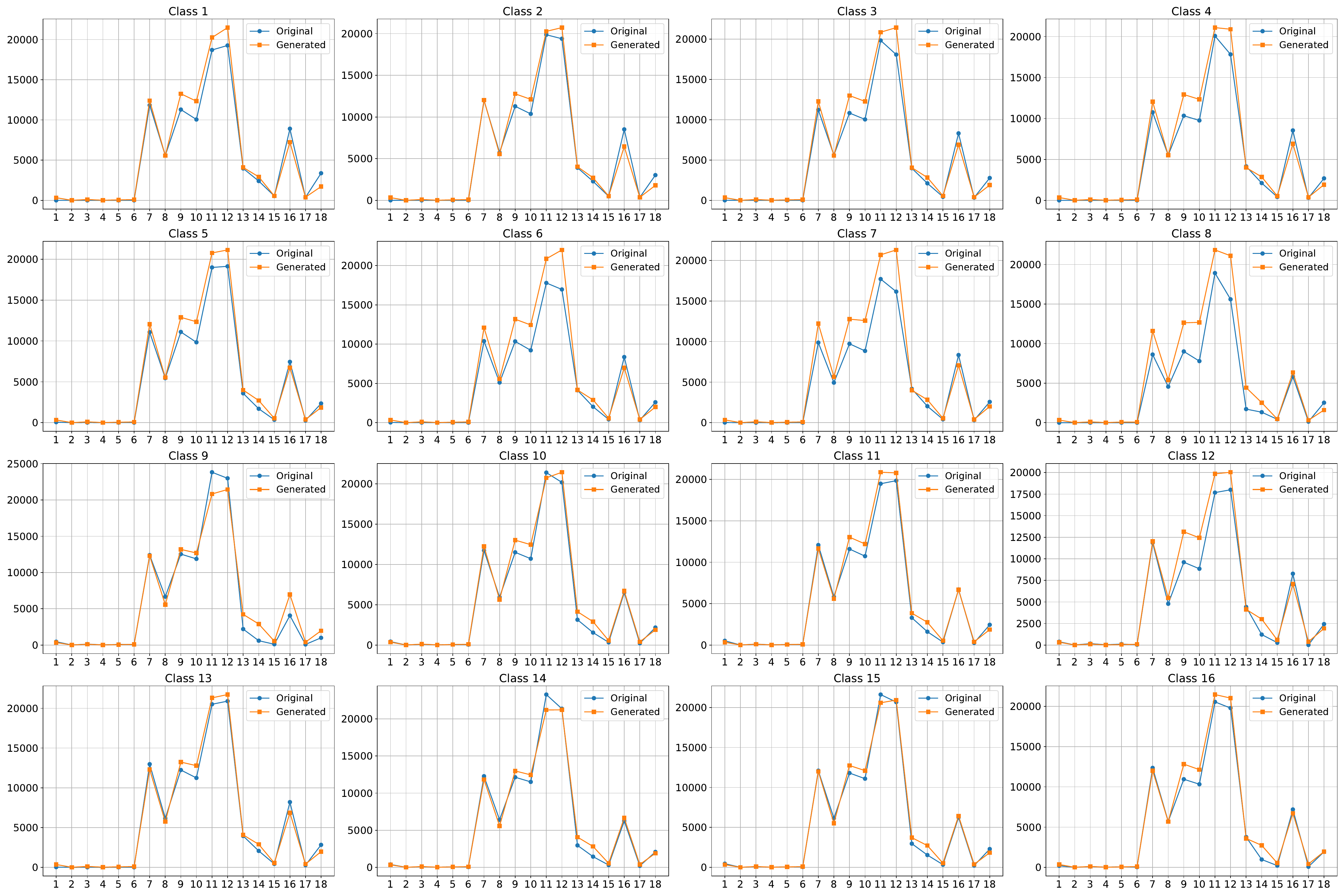}
    \caption{Comparison of averaged generation between original data and generated data in the 16 classes, we use the x-axis indices 1–6 for Solar, 7–12 for Load, and 13–18 for Wind.}
    \label{fig:average}
\end{figure*}
\paragraph{Ablation Study}
To disentangle the contributions of different training objectives in Diff2SP, we conduct an ablation study focused on the impact of the optimization-aware loss. Specifically, we aim to assess how incorporating optimization objectives into the generative process influences the quality of the generated scenarios and their effectiveness in downstream decision-making.

{
To assess the contribution of the optimization-aware training objective, we conduct an ablation study comparing three variants of the diffusion-based scenario generator. The vanilla diffusion model is trained using only the standard noise-prediction loss. The second variant adds the reconstruction loss to improve sample fidelity but omits the optimization-aware term. The full Diff2SP model includes all three loss components, thereby combining noise prediction, reconstruction accuracy, and downstream decision alignment in a single training objective.
}



Since we applied data normalization prior to training, we maintain consistency by setting the weights of all loss components $\lambda_{\text{noise}}, \lambda_{\text{recon}}, \lambda_{\text{opt}}$ to 1 in the Full Diff2SP configuration. This ensures that each loss term contributes equally, preventing any component from dominating the training process while preserving the model's ability to generate high-quality, optimization-aligned scenarios. {Besides, to evaluate whether the generated scenarios preserve the features that matter for stochastic power-system optimization, we compare the historical and generated scenario sets within each scenario class $c$ defined in Section \ref{sec:datasetup}. Let $\xi_c$ denote the historical uncertainty realizations in class $c$, and let $\hat{\xi}_c$ denote the corresponding generated realizations. Since scenario quality in SP cannot be assessed solely by visual or marginal distributional similarity, we evaluate each method along three complementary dimensions: its impact on the downstream optimization objective, its ability to reproduce the first- and second-order statistical structure of the uncertainty, and its ability to preserve constraint-relevant operating behavior in the DC-OPF model.

\noindent We first evaluate whether the generated scenarios preserve the optimization implications of the historical realizations within each scenario class. Let $\xi_c$ denote the collection of historical uncertainty realizations belonging to class $c$, and let $\hat{\xi}_c$ denote the corresponding generated realizations for the same class. Let $x^\star$ be the first-stage decision obtained by solving the SP problem using the historical realizations in class $c$. Holding $x^\star$ fixed, we compare its objective value under $\xi_c$ and $\hat{\xi}_c$. We define the normalized objective deviation as
\begin{equation}\label{err1}
\mathrm{err}^{\mathrm{obj}}_c
=
\frac{\big|Q(x^\star;\hat{\xi}_c)-Q(x^\star;\xi_c)\big|}
{\big|Q(x^\star;\hat{\xi}_c)\big|+\big|Q(x^\star;\xi_c)\big|}.
\end{equation}
This normalization places the two objective evaluations on a common scale and captures the relative discrepancy induced by replacing the historical scenario set with the generated one under the same fixed decision, avoiding distortions caused by differences in objective-value magnitudes.

We then evaluate distributional fidelity through the first two empirical moments of each scenario class. Let
$\mu_c := \frac{1}{|\xi_c|}\sum_{\omega\in\xi_c}\omega$ and
$\hat{\mu}_c := \frac{1}{|\hat{\xi}_c|}\sum_{\omega\in\hat{\xi}_c}\omega$
denote the empirical mean vectors of the historical and generated scenarios, respectively. The relative mean deviation is given by
\begin{equation}\label{err2}
\mathrm{err}^{\mu}_c=\frac{\|\hat{\mu}_c-\mu_c\|_2}{\|\mu_c\|_2}.
\end{equation}
Since each uncertainty realization contains wind, solar, and load information across all zones, this metric measures whether the generated scenarios reproduce the average operating conditions within each class.

Beyond the mean comparison, we also examine whether the generated scenarios preserve variability and dependence patterns.
Let $\Sigma_c = \frac{1}{|\xi_c|-1}\sum_{\omega\in\xi_c} (\omega-\mu_c)(\omega-\mu_c)^\top$ and $\hat{\Sigma}_c=\frac{1}{|\hat{\xi}_c|-1}\sum_{\omega\in\hat{\xi}_c}(\omega-\hat{\mu}_c)(\omega-\hat{\mu}_c)^\top$ denote the empirical covariance matrices of the historical and generated scenarios, respectively. We then define the covariance deviation as
\begin{equation}\label{err3}
\mathrm{err}^{\Sigma}_c
=
\frac{\|\hat{\Sigma}_c-\Sigma_c\|_F}{\|\Sigma_c\|_F}.
\end{equation}
The covariance deviation complements the mean deviation by measuring how well the generated scenarios preserve the second-order dependence structure, including temporal, cross-zone, and cross-variable correlations among wind, solar, and load.

Finally, we assess whether the generated scenarios preserve the operational stress imposed on transmission constraints. In the DC-OPF model, line-limit violations and the associated slack variables reflect constraint-relevant behavior that may not be captured by moment-based statistics alone. For a scenario set $\mathcal{S}$, define the average total line-limit violation as $\mathrm{LS}(\mathcal{S})
:=
\mathbb{E}_{s\in \mathcal{S}}
\left[
\sum_{(i,j)\in E}
\left(
[f_{ij}^s-F_{ij}]_+
+
[-F_{ij}-f_{ij}^s]_+
\right)
\right]$,
where $f_{ij}^s$ is the line flow on edge $(i,j)$ under scenario $s$, $F_{ij}$ is the corresponding line-flow limit, and $[z]_+ := \max\{z,0\}$. We then compare the line-slack behavior of the generated and historical scenarios through
\begin{equation}\label{err4}
\mathrm{err}^{\mathrm{line}}_c
=
\frac{\big|\mathrm{LS}(\hat{\xi}_c)-\mathrm{LS}(\xi_c)\big|}
{\mathrm{LS}(\xi_c)}.
\end{equation}
This metric measures whether the generated scenarios reproduce the severity of transmission constraint violations observed under the historical scenarios.
}

\begin{figure}[ht]
    \centering
    \includegraphics[width=0.93\linewidth]{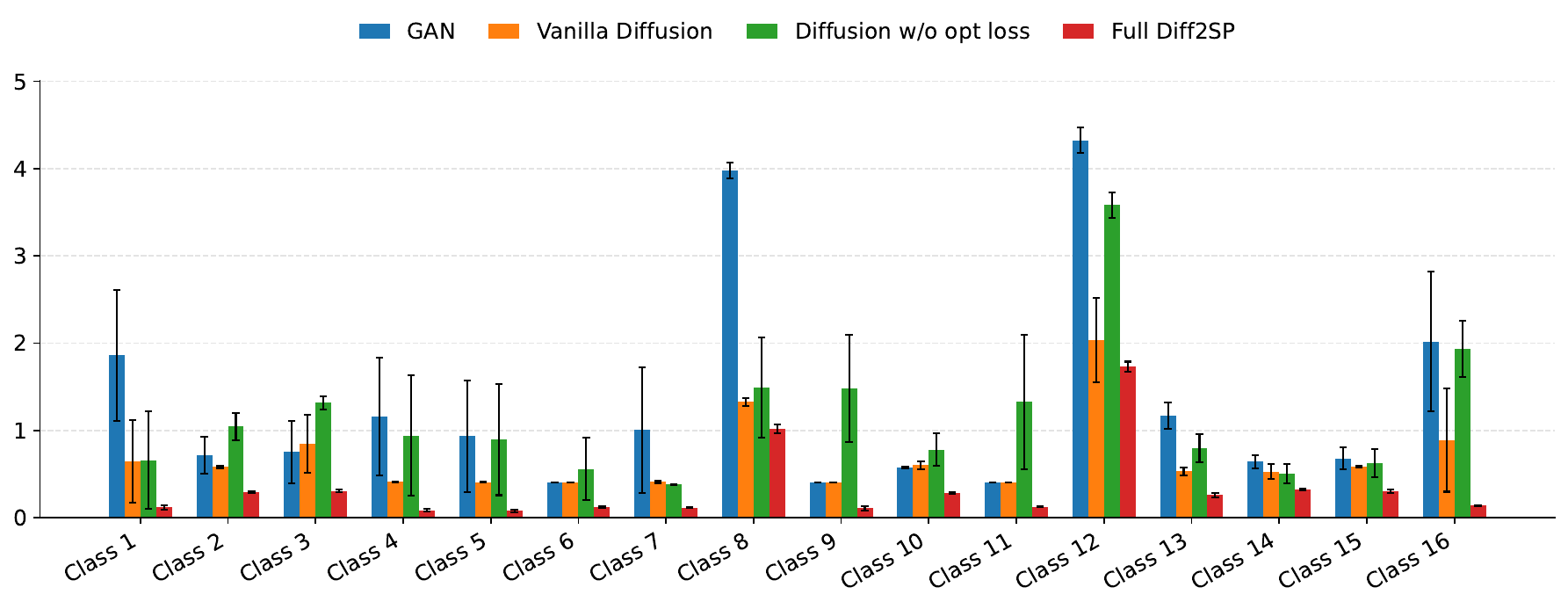}
    \caption{Performance Comparison of Different Diffusion Model Variants Across Scenario Classes. The y-axis reports the average of four error, as defined in (Eq. \eqref{err1} - Eq. \eqref{err4}).}
    \label{fig:res}
\end{figure}
Figure~\ref{fig:res} summarizes class-level fidelity using a composite score, where each bar reports the average of above four class-wise deviation metrics. To prevent the covariance and line-slack terms from dominating due to their larger dynamic range, we apply a logarithmic transform to these two components in the composite score. Full Diff2SP achieves the lowest composite error across all scenario classes, indicating the most reliable alignment between generated scenarios and downstream OPF behavior at the class level. The ablation results pinpoint the key driver, removing the optimization-guided loss leads to consistently larger errors across classes, showing that purely statistical objectives do not preserve decision-critical and constraint-relevant structure. In contrast, optimization-guided training enforces decision-aware scenario generation, yielding uniformly low deviations across all classes.

\section{Conclusion}
\label{sec_conclusion}
In this work, we proposed Diff2SP, a diffusion-based scenario generation framework for stochastic programming that learns high-fidelity uncertainty distributions without supervision. Unlike traditional Monte Carlo or GAN-based approaches, Diff2SP preserves both marginal statistics and cross-variable correlations by integrating self-attention mechanisms into the diffusion process. This structure enables the model to capture complex temporal and multivariate dependencies, leading to scenario sets that enhance decision quality in downstream optimization tasks. 

Empirical results show that Diff2SP consistently outperforms baseline methods in producing diverse, realistic, and structurally coherent scenarios. Its ability to maintain critical correlation patterns improves decision quality in applications such as energy scheduling and financial risk assessment. Moreover, by eliminating the need for labeled data or predefined scenario mappings, Diff2SP remains scalable and broadly applicable to high-dimensional uncertainty modeling.

Despite these advantages, Diff2SP currently relies on one-hot classification labels for conditional diffusion, limiting its capacity to model richer contextual information. Future work will explore multi-modal extensions, integrating textual descriptions and structured metadata to improve interpretability and generalization across diverse decision-making applications. 

\clearpage

%
%
%
\appendix

\section{Enhancing Fidelity with Transformers}
\label{subsec:Transformer}
To better capture long-range temporal and cross-variable dependencies in the reverse diffusion process, we replace the conventional CNN-based parameterization of the reverse mean function \( \mu_\theta(\cdot) \) with a Transformer-based architecture \citep{vaswani2017attention}. Concretely, at each diffusion step \( k \), the mean of the reverse Gaussian in \eqref{eq:reverse} is given by
\[
\mu_\theta(\xi^k, c, k) \;=\; \mathcal{T}_\theta(\xi^k, c, k),
\]
where \( \mathcal{T}_\theta : \mathbb{R}^{T \times D} \times \mathcal{C} \to \mathbb{R}^{T \times D} \) is a neural operator implemented by a Transformer network and \( T \) is the number of physical time periods in the time series.

To avoid confusion, we explicitly distinguish two notions of time here: the \emph{diffusion step} \( k = 1, \ldots, K \), which indexes the noise level in the denoising process, and the \emph{real-world time index} \( t = 1, \ldots, T \), which indexes the temporal resolution within each multivariate time series. At each diffusion step \( k \), the noisy sample
\[
\xi^k = 
\begin{bmatrix}
(\xi^k_1)^\top \\
\vdots \\
(\xi^k_T)^\top
\end{bmatrix}
\in \mathbb{R}^{T \times D}
\]
represents a multivariate time series of length \( T \), where each row
$
\xi^k_t \in \mathbb{R}^D
$
collects the values of the \( D \) uncertain quantities (such as wind and solar generation, and energy demand) at physical time \( t \). 

The Transformer operates on a sequence of vector-valued inputs, often referred to as \emph{tokens}. In our setting, each token corresponds to one time step \( t \) of the multivariate series. To map the raw inputs into a space where the Transformer can more easily learn abstract patterns, we first apply a linear embedding to each time step:
\begin{equation}
\label{eq:MLP}
Z_t^{(0)} = W_e \, \xi^k_t + b_e, 
\quad t = 1, \ldots, T,
\end{equation}
where \( W_e \in \mathbb{R}^{D \times d_{\text{model}}} \) and
\( b_e \in \mathbb{R}^{d_{\text{model}}} \) are trainable parameters, and
\( d_{\text{model}} \) is a user-defined hidden dimension that controls the representational capacity of the Transformer. The subscript \( e \) indicates
that this layer serves as an \emph{embedding} from the original feature space
into the Transformer hidden space. The superscript \( (0) \) denotes the input layer of the Transformer. Each vector \( Z_t^{(0)} \in \mathbb{R}^{d_{\text{model}}} \) is the embedded representation of the system state at time \( t \).

To encode temporal order and external context, we augment each embedded vector \( Z_t^{(0)} \) with two additional components. First, a \emph{positional encoding} \( PE_t \in \mathbb{R}^{d_{\text{model}}} \) encodes the time index (such as hour of the day). Second, a \emph{contextual encoding} \( CE(c) \in \mathbb{R}^{d_{\text{model}}} \) embeds observable exogenous covariates \( c \) such as weather forecasts or scenario descriptors (for example, high-renewable or peak-demand conditions). The final input token at time \( t \) is
\begin{equation}
Z_t = Z_t^{(0)} + PE_t + CE(c),
\quad t = 1, \ldots, T,
\end{equation}
so that the Transformer receives, for each time step, a vector that encodes its feature content, temporal position, and external context.

The core mechanism in the Transformer is \emph{self-attention}, which allows each token to aggregate information from all other tokens in the sequence. Intuitively, when updating the representation at time \( t \), the model can ``look at'' all time steps \( t' = 1, \ldots, T \) and decide how much each of them should contribute.
Specifically, for each token \( Z_t \), the Transformer computes three vectors via linear projections:
\[
Q_t = Z_t W_Q, \qquad
K_{t'} = Z_{t'} W_K, \qquad
V_{t'} = Z_{t'} W_V,
\]
where \( W_Q, W_K, W_V \in \mathbb{R}^{d_{\text{model}} \times d_h} \) are trainable matrices and \( d_h \) is the dimensionality of the attention subspace. The vectors \( Q_t \), \( K_{t'} \), and \( V_{t'} \) are referred to as the \emph{query}, \emph{key}, and \emph{value}, respectively.

The similarity between time \( t \) and \( t' \) is measured by the scaled dot product
\[
a_{tt'} = \frac{Q_t K_{t'}^\top}{\sqrt{d_k}},
\]
which is then normalized over \( t' \) using a softmax to obtain attention weights
\[
A_{tt'} =
\frac{\exp(a_{tt'})}
{\sum_{j=1}^{T} \exp(a_{tj})},
\quad t, t' = 1, \ldots, T.
\]
The updated representation at time \( t \) is a weighted average of all value vectors:
\begin{equation}
\label{eq:Attention}
\widehat{Z}_t = \sum_{t'=1}^{T} A_{tt'} V_{t'}, \quad t = 1, \ldots, T.
\end{equation}
Thus, \eqref{eq:Attention} allows each time index \( t \) to selectively aggregate information from all other times based on learned relevance weights \( A_{tt'} \), capturing long-range temporal and cross-variable dependencies.

In practice, the Transformer employs \emph{multi-head attention}. Instead of computing attention once in the full \( d_{\text{model}} \)-dimensional space, the model splits the representation into \( h \) subspaces of dimension \( d_h \), such that
$
d_{\text{model}} = h \, d_h.
$
Each head has its own triplet \( (W_Q^{(h)}, W_K^{(h)}, W_V^{(h)}) \) and computes an update of the form \eqref{eq:Attention} in its own subspace. The outputs from all heads are then concatenated and passed through a linear projection to form the final updated token representation. Stacking several such attention-plus-feedforward layers yields the overall operator \( \mathcal{T}_\theta \) that maps the noisy input sequence \( \xi^k \) (together with context \( c \)) to the denoised mean \( \mu_\theta(\xi^k, c, k) \).

By enabling direct information exchange across all time steps and variables at each denoising stage, the Transformer-based reverse model \( \mathcal{T}_\theta \) preserves long-range temporal structure and joint dependence in the generated scenarios. While this architectural enhancement improves structural fidelity, it does not by itself guarantee downstream decision relevance. In the next subsection, we introduce an optimization-guided training strategy that explicitly aligns scenario generation with SP objectives.

\section{Proof of Lemma \ref{lem:RegretBd}}
\label{sec:App_RegretBd}
\begin{proof}{Proof.} For notational convenience, we write $ \mathbb{E}_{P} $ and $ \mathbb{E}_{\mathcal{Q}_\theta} $ to denote $ \mathbb{E}_{\xi \sim\mathcal{P}} $ and $ \mathbb{E}_{\xi \sim \mathcal{Q}_\theta} $, respectively, throughout the paper.
Based on the regret definition $R(\theta)$ in \eqref{eq:regret}, we can add and subtract $ \mathbb{E}_{\xi \sim \mathcal{Q}_\theta}[f(\txth, \xi)] $ and $ \mathbb{E}_{\xi \sim \mathcal{Q}_\theta}[f( x^*, \xi)] $ to obtain:
	\begin{align}
	R(\theta) = & \ \underbrace{\left( \mathbb{E}_{P}[f(\txth, \xi)] - \mathbb{E}_{\mathcal{Q}_\theta}[f(\txth, \xi)] \right)}_{(A)} \nonumber\\
	& + \underbrace{\left( \mathbb{E}_{\mathcal{Q}_\theta}[f(\txth, \xi)] - \mathbb{E}_{\mathcal{Q}_\theta}[f( x^*, \xi)] \right)}_{(B)} \label{eq:Regret_expand}\\
	& + \underbrace{\left( \mathbb{E}_{\mathcal{Q}_\theta}[f( x^*, \xi)] - \mathbb{E}_{P}[f( x^*, \xi)] \right)}_{(C)}. \nonumber
	\end{align}
For Term (A), since $ f( x, \xi) $ is $ L $-Lipschitz in $ \xi $, we have
\begin{equation*}
\mathbb{E}_{P}[f(\txth, \xi)] - \mathbb{E}_{\mathcal{Q}_\theta}[f(\txth, \xi)] \le L \cdot W_1(P, \mathcal{Q}_\theta)
\end{equation*}
by the Kantorovich–Rubinstein formula \citep[Particular Case 5.16]{OptimalTransport}, which states that for any two probability measures $ \mu $ and $ \nu $ on a metric space $ (\Omega, d) $ with finite first moments, the 1-Wasserstein distance satisfies:
$
W_1(\mu, \nu)
= \sup_{\phi \in \mathrm{Lip}_1}
\left\{
\mathbb{E}_{\xi \sim \mu}[\phi( \xi)] - \mathbb{E}_{\xi \sim \nu}[\phi( \xi)]
\right\},
$
where $ \mathrm{Lip}_1 $ denotes the set of all 1-Lipschitz functions $ \phi : \Omega \to \mathbb{R} $.
For Term (C), the same upper bound can be obtained using the same argument. For Term (B), the optimality of $\txth$ under the distribution $\mathcal{Q}_\theta$ implies that $\mathbb{E}_{\mathcal{Q}_\theta}[f(\txth, \xi)] - \mathbb{E}_{\mathcal{Q}_\theta}[f( x^*, \xi)] \leq 0$. Hence, we have that
$R(\theta) \le  L \cdot W_1(P, \mathcal{Q}_\theta) + 0 + L \cdot W_1(P, \mathcal{Q}_\theta) = 2L \cdot W_1(P, \mathcal{Q}_\theta).$ 
\end{proof}

\section{Proof of Theorem \ref{thm:diff2sp-bound}}
\label{app:thm2}
\begin{proof}{Proof.}
We first decompose the error into the estimation error and the approximation error.
\[
R(\hat{\theta})
= \underbrace{L(\hat{\theta}) - L(\theta^\star)}_{\text{estimation}}
+
\underbrace{L(\theta^\star) - L(x^*)}_{\text{approximation}}.
\]
The approximation term is nonnegative and vanishes if $x^* \in \{x_\theta:\theta\in\Theta\}$.


A standard uniform deviation bound (symmetrization + McDiarmid) then yields,  with probability at least $1-\delta/2$,

\begin{equation*}
\begin{aligned}
\sup_{\theta\in \Theta}\big|L(x_{\theta})-\widehat{L}_N(x_{\theta})\big|
&\le 2\mathcal{R}_N + LD\sqrt{\tfrac{\log(2/\delta)}{2N}}
\\&\le \frac{2C}{\sqrt{N}} + LD\sqrt{\tfrac{\log(2/\delta)}{2N}}.
\end{aligned}
\end{equation*}

Since $\hat{\theta}$ minimizes $\hat{L}_N$, we have
\begin{equation*}
        L(\hat{\theta}) - L(\theta^\star)
\le \underbrace{\left(L(\hat{\theta}) - \hat{L}_N(\hat{\theta})\right)}_{\le \sup |L-\hat L_N|}
+ \underbrace{\left(\hat{L}_N(\hat{\theta}) - \hat{L}_N(\theta^\star)\right)}_{\le 0}+ \underbrace{\left(\hat{L}_N(\theta^\star) - L(\theta^\star)\right)}_{\le \sup |L-\hat L_N|},
\end{equation*}
so that
\begin{equation*}
    \begin{aligned}
L(\hat{\theta}) - L(\theta^\star)
&\le 2 \sup_{\theta}\big|L(\theta)-\hat{L}_N(\theta)\big|\\
&\le \frac{4C}{\sqrt{N}} + 2LD\sqrt{\tfrac{\log(2/\delta)}{2N}}.
    \end{aligned}
\end{equation*}

With probability at least $1-\delta$,
\[
R(\hat{\theta})
\le
\frac{4C}{\sqrt{N}}
+ 2LD\sqrt{\tfrac{\log(2/\delta)}{2N}}
+ \big(L(\theta^\star)-L(x^*)\big).
\]

To require the estimation error to be bounded by $\varepsilon$ is to have 
$$\frac{4C}{\sqrt{N}} + 2LD\sqrt{\tfrac{\log(2/\delta)}{2N}} \le \varepsilon,$$
which implies that 
$$\frac{1}{\sqrt{N}} \left[ 4C + LD \cdot \sqrt{2\log\left(\frac{2}{\delta}\right)} \right] \le \varepsilon.$$

Thus, for realizable models ($L(\theta^\star)=L(x^*)$) ,
it suffices to take
$$N \ge \frac{1}{\varepsilon^2} \left[ 4C + LD\sqrt{2\log\left(\frac{2}{\delta}\right)} \right]^2.$$
for an constant $C$, which yields
$\Pr\big[R(\hat{\theta})\le \varepsilon\big]\ge 1-\delta$.
\end{proof}
\begin{remark}[Why the proof does not extend to GANs]
The key step in the proof of Theorem~\ref{thm:diff2sp-bound} is Step~2, 
where we control the deviation 
$\sup_{\theta \in \Theta} |L(x_\theta) - \hat L_N(x_\theta)|$ 
via Rademacher complexity of the hypothesis class 
$\mathcal{F}=\{f(x_\theta,\cdot):\theta\in\Theta\}$. 
This works in Diff2SP because the learning objective is an 
empirical risk functional $\widehat L_N(x_\theta)$: 
the induced decision $x_\theta$ is well-defined, and its true performance 
$L(x_\theta)$ can be related to the empirical counterpart through standard uniform convergence arguments. 

In contrast, a GAN generator does not minimize $\widehat L_N(x_\theta)$ directly. 
Instead, it is trained through an adversarial min–max game against a discriminator, 
optimizing a statistical divergence such as the Jensen–Shannon divergence. 
Thus, the generator’s objective cannot be written as a uniform empirical average 
of the form $\hat L_N(x_\theta)$, and Step~2 of the above proof does not apply. 
What can be bounded for GANs is the discriminator’s generalization error 
(see Arora et al.\ 2017), which controls divergence between $\mathcal{Q}_\theta$ and $\mathcal{P}$, 
but only indirectly relates to the decision risk $L(x_\theta)$. 
This translation from divergence bounds to regret bounds yields weaker rates, 
typically of order $O(N^{-1/4})$ rather than the $O(N^{-1/2})$ rate obtained for Diff2SP.
\end{remark}

\section{Proof of Theorem~\ref{thm:gan-bound}}
\label{app:thm3}
\begin{proof}{Proof}
Let $\hat{P}$ denote the empirical distribution formed by $N$ i.i.d.\ samples from the true distribution $\mathcal{P}$. The approximation error of the generator class is defined as
\[
\varepsilon_0 := \inf_{\theta \in \Theta} \mathrm{JS}(\hat{P}, \mathcal{Q}_\theta),
\]
where $\mathrm{JS}$ denotes the Jensen--Shannon divergence. Suppose the generator is trained to $\epsilon_{\mathrm{opt}}$-optimality with respect to the empirical objective, that is,
\[
\mathrm{JS}(\hat{P}, Q_{\hat{\theta}}) \le \varepsilon_0 + \epsilon_{\mathrm{opt}}.
\]

By the generalization theorem of \citet[Theorem~3.1]{arora2017generalization}, for a discriminator class of VC dimension $d$, the Jensen--Shannon divergence between the empirical and population distributions satisfies
\[
\big| \mathrm{JS}(P, Q_{\hat{\theta}}) - \mathrm{JS}(\hat{P}, Q_{\hat{\theta}}) \big|
\le C_1 \sqrt{\frac{d \log(A / \varepsilon)}{N}},
\]
with probability at least $1 - \delta/2$, where $A=L_{\mathcal{D}}L_{\phi}p/\delta$, with $L_{\mathcal{D}}$ denoting the
Lipschitz constant of the discriminator class with respect to its parameters,
$L_{\phi}$ denoting the Lipschitz constant of the measuring function, and $p$
denoting the number of discriminator parameters.

Combining the empirical optimality condition with this generalization result yields
\[
\mathrm{JS}(P, Q_{\hat{\theta}})
\le \varepsilon_0 + \epsilon_{\mathrm{opt}} + C_1 \sqrt{\frac{d \log(A / \varepsilon)}{N}}.
\]

On a compact support $\Xi$ with diameter $B$, the total variation distance satisfies
$W_1(P, Q) \le B \, \mathrm{TV}(P, Q)$
\citep[Particular Case 6.16]{OptimalTransport},
and by Pinsker’s inequality
$\mathrm{TV}(P, Q) \le \sqrt{2\,\mathrm{JS}(P, Q)}$
\citep[Eq.~(1)]{sason2016fdivergence}. Hence
\[
W_1(P, Q_{\hat{\theta}}) \le B \sqrt{2\,\mathrm{JS}(P, Q_{\hat{\theta}})}.
\]
Using Lemma~\ref{lem:RegretBd}, which gives $R(\theta) \le 2L\, W_1(P, \mathcal{Q}_\theta)$,
we obtain
\[
R(\hat{\theta}) \le 2 L B \sqrt{2\!\left(\varepsilon_0 + \epsilon_{\mathrm{opt}}
+ C_1 \sqrt{\frac{d \log(A / \varepsilon)}{N}}\right)}.
\]

To guarantee $R(\hat{\theta}) \le \varepsilon$, it suffices that
\[
2LB \sqrt{2\!\left(\varepsilon_0 + \epsilon_{\mathrm{opt}}
+ C_1 \sqrt{\frac{d \log(A / \varepsilon)}{N}}\right)} \le \varepsilon.
\]
Squaring both sides and rearranging terms gives
\[
8L^2 B^2 \!\left(\varepsilon_0 + \epsilon_{\mathrm{opt}}
+ C_1 \sqrt{\frac{d \log(A / \varepsilon)}{N}}\right)
\le \varepsilon^2,
\]
or equivalently,
\[
\varepsilon_0 + \epsilon_{\mathrm{opt}}
+ C_1 \sqrt{\frac{d \log(A / \varepsilon)}{N}}
\le \frac{\varepsilon^2}{8L^2 B^2}.
\]
Assuming that the empirical optimization is controlled to
$\epsilon_{\mathrm{opt}} \le \varepsilon^2 / (16L^2 B^2)$, we obtain
\[
\varepsilon_0 + C_1 \sqrt{\frac{d \log(A / \varepsilon)}{N}}
\le \frac{\varepsilon^2}{16L^2 B^2}.
\]
This condition requires $\varepsilon_0 < \varepsilon^2 / (16L^2 B^2)$; otherwise,
no sample size can achieve the desired regret bound.
Solving for $N$ gives
\[
C_1 \sqrt{\frac{d \log(A / \varepsilon)}{N}}
\le \frac{\varepsilon^2}{16L^2 B^2} - \varepsilon_0,
\qquad
N \ge
\frac{C_1^2 d \log(A / \varepsilon)}{\left(\frac{\varepsilon^2}{16L^2 B^2} - \varepsilon_0\right)^2}.
\]

Finally, applying standard VC generalization arguments introduces an additional
$\log(1/\delta)$ factor inside the logarithm, which can be absorbed into $A$.
Thus, with probability at least $1-\delta$,
\[
N \ge
\frac{C_1^2 d\,\log\!\big(3 D_\Xi / (\varepsilon \delta)\big)}
{\left(\frac{\varepsilon^2}{16L^2 B^2} - \varepsilon_0\right)^2},
\]
which establishes the claimed sample–complexity dependence. 
\end{proof}




\bibliographystyle{plainnat} 
\bibliography{sample} 





\end{document}